\documentclass[11pt,english]{article}
\usepackage[sc]{mathpazo}
\usepackage{geometry}
\usepackage{color}
\usepackage{array}
\usepackage{bbding}
\usepackage{multirow}
\usepackage[fleqn]{amsmath}
\usepackage{amssymb}
\usepackage{amsthm}
\usepackage{graphicx}
\usepackage{natbib}
\usepackage{lscape}
\usepackage{comment}
\usepackage{bm}
\usepackage[title]{appendix}
\usepackage{longtable}
\usepackage{mathtools}
\usepackage{chngcntr}
\usepackage{authblk}
\usepackage{enumerate}
\makeatletter
\allowdisplaybreaks

\usepackage[colorlinks,
            linkcolor=blue,
            anchorcolor=blue,
            citecolor=blue
            ]{hyperref}

\@ifundefined{showcaptionsetup}{}{%
 \PassOptionsToPackage{caption=false}{subfig}}
\usepackage{subfig}
\makeatother

\usepackage{babel}
\usepackage{enumitem}

\newtheorem{prop}{Proposition}

\begin{document}\sloppy

\nocite{*}
\title{Aggregate Modeling and Equilibrium Analysis of the Crowdsourcing Market for Autonomous Vehicles}
\author{Xiaoyan Wang\textsuperscript{a}\hspace{1em} Xi Lin\textsuperscript{a}\hspace{1em} Meng Li\textsuperscript{a,}\textsuperscript{b,}\footnote{Corresponding author. E-mail address: \textcolor{blue}{mengli@tsinghua.edu.cn}.}}
\affil{\small\emph{\textsuperscript{a}Department of Civil Engineering, Tsinghua University, Beijing 100084, P.R. China}\normalsize}
\affil{\small\emph{\textsuperscript{b}Center for Intelligent Connected Vehicles and Transportation, Tsinghua University, Beijing 100084, P.R. China}\normalsize}
\date{}
\maketitle

\begin{abstract}
\noindent
Autonomous vehicles (AVs) have the potential of reshaping the human mobility in a wide variety of aspects. This paper focuses on a new possibility that the AV owners have the option of "renting" their AVs to a company, which can use these collected AVs to provide on-demand ride services without any drivers. We call such a mobility market with AV renting options the "AV crowdsourcing market". This paper establishes an aggregate equilibrium model with multiple transport modes to analyze the AV crowdsourcing market. The modeling framework can capture the customers' mode choices and AV owners' rental decisions with the presence of traffic congestion. Then, we explore different scenarios that either maximize the crowdsourcing platform's profit or maximize social welfare. Gradient-based optimization algorithms are designed for solving the problems. The results obtained by numerical examples reveal the welfare enhancement and the strong profitability of the AV crowdsourcing service. However, when the crowdsourcing scale is small, the crowdsourcing platform might not be profitable. A second-best pricing scheme is able to avoid such undesirable cases. The insights generated from the analyses provide guidance for regulators, service providers and citizens to make future decisions regarding the utilization of the AV crowdsourcing markets for serving the good of the society.  \par
\hfill\break%
\end{abstract}
\section{Introduction} \label{sec1}

\noindent With the rapid development of artificial intelligence and communication technology in recent years, autonomous vehicles (AVs) are undergoing rapid development. There has been an intense effort by researchers and manufactures to enable AVs to handle a wide range of driving conditions that can be encountered  during the road travel. In 2018, public road tests for AVs to provide on-demand ride service were performed, and in 2020, some major transportation network companies (TNCs) such as Lyft, Waymo and Baidu have launched commercialized AV on-demand ride service programs in certain areas around the globe. It is predicted that approximately $40\%$ of vehicle travel could be autonomous in the 2040s and traveling by AV may become the predominant transport mode in the future. It is widely believed that due to the increased mobility and safety, reduced congestion, emissions and travel costs, AVs have the potential to reshape future human transport.
 \par

High-tech AVs can perform tasks that cannot be performed by manually driven vehicles (MVs). One of the most distinct differences between AVs and MVs is that AVs owned by a person can be used for other tasks when not used by the owner such as traveling to other places without drivers to complete certain tasks. One potential use is to serve other riders (Figure \ref{illustration_fig} for illustration). This would provide opportunities for people who do not own AVs to utilize the service.

\begin{figure}[!ht]
	\centering
	\includegraphics[width=1\textwidth]{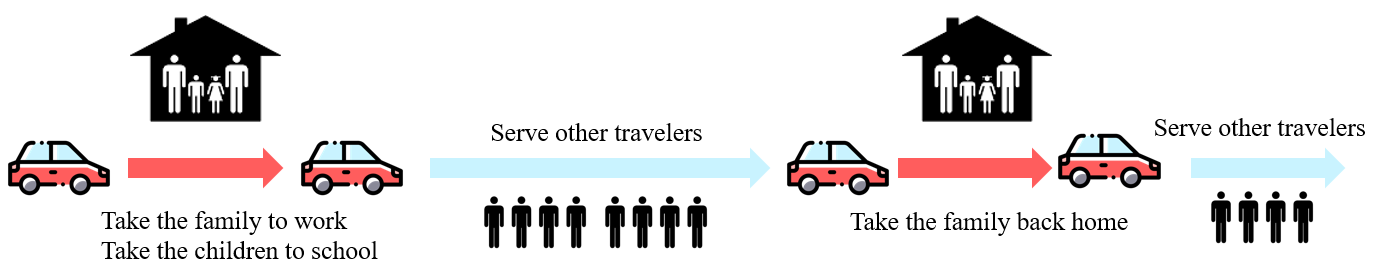}
	\caption[]{One day's task chain of an autonomous vehicle} 
	\label{illustration_fig}
\end{figure}
\par
\noindent

When a large number of AV owners are able to share their AVs to provide ride services to others, a platform is required to aggregate the information of those AVs with that of the customer demands, and act as a bridge connecting the AV owners with the customers. Similar to the current TNCs such as Uber and Lyft, the platform identifies the dynamically-updated supply from AV owners, and then matches the available AVs with the customers generated within the region in a real-time fashion. The platform needs to pay these AV owners a certain amount of money to encourage them to share their AVs, and meanwhile the platform charges the customers a certain fee for their trips; the difference between the platform gain from travelers and the platform payment to AV owners is the revenue that should be equal to a considerable amount so that people have the incentive to establish and operate such a service.
\noindent

The aforementioned new type of mobility market with AV renting options is called "the AV crowdsourcing market" in this study. The concept of crowdsourcing, first introduced by Brabham in 2008  and later perfected by other researchers using case studies, refers to the act of a company handing over a part of its tasks, particularly needs-based problems, such as hatching ideas, designing algorithms and public supervision to the public crowd who have the potential to finish the job better than employees. Developed communication technologies enable such business model to be realized, through which a wide range of the public including but not limited to social experts are attracted and involved to produce a better outcome for the companies. For the AV crowdsourcing market, our primary interest is to answer two key questions: 1) can AV crowdsourcing benefit society? and 2) can the operating agency gain from operating the crowdsourcing platform? The first question relates to the utilitarian objective of such business model that tells social managers whether they should allow or regulate the market, and the second question is related to the spontaneity problem whether people have the incentive to run such a business. To answer these questions, we establish an aggregate equilibrium modeling framework in this study. The model considers three market subjects: travelers, AV owners and the crowdsourcing platform, where AV owners and travelers may have some coincidence within a time period. The platform determines the payment to the AV owners and the charge to the travelers, and the travelers and AV owners determine their mode choices and rental choices. Under a user's utility-maximization framework, we can capture the equilibrated state of the whole system. Then, based on the equilibrium state, the platform can adjust the payment and pricing strategies to achieve different goals, including profit maximization, social welfare maximization (first-best), and constrained social welfare maximization (second-best). Numerical experiments and sensitivity analyses are conducted to draw a variety of insights on the AV crowdsourcing market.
\noindent

The rest of this paper is organized as follows. Section \ref{sec2} reviews related literature on AVs, ride-providing services and crowdsourcing. In Section \ref{sec3}, we develop the static equilibrium model for the AV crowdsourcing market, and investigate some properties of the model. Section \ref{sec4} formulates the platform's decision framework across different market scenarios with various transport modes, and then propose a gradient-based algorithm to solve the problem. In Section \ref{sec5}, we use numerical examples to illustrate key points and insights of our model, and then present some discussions. Finally, Section \ref{sec6} concludes the paper.
\par
\section{Related literature} \label{sec2}

The body of literature on AVs has been continuously growing in recent years. Control of fully AVs may be the best investigated topic (Wu et al., 2020), including route planning (Wang et al., 2019; Bang and Ahn, 2018), longitudinal control, i.e. Cooperative Adaptive Cruise Control (Wang et al., 2014; Gong and Du, 2018), and lateral control of autonomous vehicles (Luo et al., 2016; Yang et al., 2018). As a sophisticated scenario in transportation ﬁeld, control of an isolated intersection has been studied with "signal-free" schemes (Lee and Park, 2012; Xu et al., 2018) and "signalized" schemes (Li et al., 2014; Yu  et al., 2018) where the aim is to increase intersection capacity and reduce delay on the basis of driving safety. Recently, Chen et al. (2020) proposed a rhythmic control scheme where CAVs pass through an intersection with a preset rhythm, and Lin et al. (2021), further expended the rhythmic control to a grid network. Meanwhile, studies found that AVs with communication technologies that can smooth oscillations and reduce crash probability, will greatly improve the capacity of a highway (Michael et al., 1998; Tientrakool et al., 2011; Bian et al., 2019) and reduce road congestion (Fagnant and Kockelman, 2015; Sun et al., 2020). AVs are highly advantageous for providing high-level transport experience so that they will change or even reshape future travel modes. One change will revolutionize the taxi industry. Literature on AV on-demand ride service mainly focuses on its operation. For example, Vosooghi et al. (2019), considering multi-modal dynamic demand, investigated SAV operations including the ﬂeet size, vehicle capacities, ridesharing and rebalancing strategies through simulations. Tang et al. (2020) proposed an advisor-student reinforcement learning framework to organize the autonomous electric taxi ﬂeet in an online manner.
\noindent

Although there is scarcely any literature on AV on-demand ride service, a wide range of studies on the traditional taxi market and on-demand ride service market have been performed that provide great inspiration for our study. Yang and Wong (1998) used a network model to establish the equilibrium of the cruising taxi market and offered some policy-relevant results including average taxi utilization and average customer waiting time for decision making. Daganzo and Ouyang (2019) presented a general analytic framework to model transit systems that provide door-to-door service, including non-shared taxi and shared taxi service. Studies pointed out that both the traditional taxi market and on-demand ride market has similar economic rules. Arnott (1996) performed an economic analysis of the dispatch taxi market and found that the ﬁrst-best taxi pricing entails operation at a loss under economies of density. Zha et al. (2016) using an aggregate model of the on-demand ride market, pointed out that under economies of scale, the company will suffer a deﬁcit in the ﬁrst-best scenario. Vignon and Yin (2020) considered the dual modes of ride-sourcing and ride-sharing and also pointed out that when congestion is low, the service provider must be subsidized to achieve the ﬁrst-best scenario. Vignon and Yin (2020) developed several concrete approaches to achieve the second-best scenario, including regulation of prices, vehicles' vacant duration and the number of drivers. In this paper, we will show that in some circumstances, the AV crowdsourcing platform can be proﬁtable in the ﬁrst-best scenario. Other inﬂuential works on the regulation of the MV on-demand ride market include Yu et al. (2017) that analyzed the regulation of a market with both traditional taxis and on-demand ride service, and further proved that on-demand ride service is competitive with the traditional taxi service. He et al. (2018) modeled the effects of the customers' reservation cancellation behaviors on the network and designed the pricing and penalty strategy. Wang and Yang (2019) have made a comprehensive review of the literature on on-demand ride services. However, none of the above works considered the future transport modes in which AVs can drive autonomously and can serve customers on their own (Zmud et al., 2018; Stocker and Shaheen, 2018; Narayanan et al., 2020). The potential AV crowdsourcing market enabled by the driverless feature of AVs has been seldom investigated in the existing literature.
\noindent

Crowdsourcing has boomed in various industries in recent years (Brabham, 2008). It can be applied to collect ideas to improve products or services (Bayus, 2010; Poetz and Schreier, 2012; Bayus, 2013), city inspection (Kang et al., 2013; Glaeser et al., 2016) and other applications. Platforms such as Threadless and IStockPhoto collect and share designer's talents around the world; Kaggle and TunedIT gather studies and solutions to difﬁcult problems; other companies such as PeoplePerHour, Crowdspirit, Editzen and MusikPitch are examples of crowdsourcing (Zhao and Zhu, 2014). A business model in the hotel industry similar to AV crowdsourcing is shared lodging where one of the typical companies is Airbnb (Zervas et al., 2017). The AV crowdsourcing that we propose in this paper has been referred to in Stocker and Shaheen (2018) and described as hybrid AV ownership with the same entity operating. Stocker and Shaheen (2018) preliminary listed potential shared autonomous vehicle business models but did not make a detailed analysis of them; to bridge this gap, such models are investigated in the present paper.

\section{Equilibrium Model} \label{sec3}

This section introduces the equilibrium model of the AV crowdsourcing market. We consider a time period with duration $h$ (e.g., 2 hours) for which the market conditions are time invariant, and therefore we can treat it as a stationary state. In this hypothetical market, there exists a crowdsourcing platform; the AV owners have the option to rent their AVs to this crowdsourcing platform, and the crowdsourcing platform collects these AVs to provide on-demand ride services for other riders. In addition to the collected AVs, the crowdsourcing platform can also buy some AVs from external markets to provide the services. On the other hand, travelers can choose to use the on-demand ride services, to drive private (manual or autonomous) vehicles or to take public transit to fulfill their travel needs (the average trip distances of taking different modes are identical); meanwhile, AV owners can choose to rent their cars to this crowdsourcing platform. Traffic congestion induced by these trips is considered in the framework. A schematic illustration is provided in Figure
\ref{concept_fig}. 
\begin{figure}[!ht]
	\centering
	\includegraphics[width=1\textwidth]{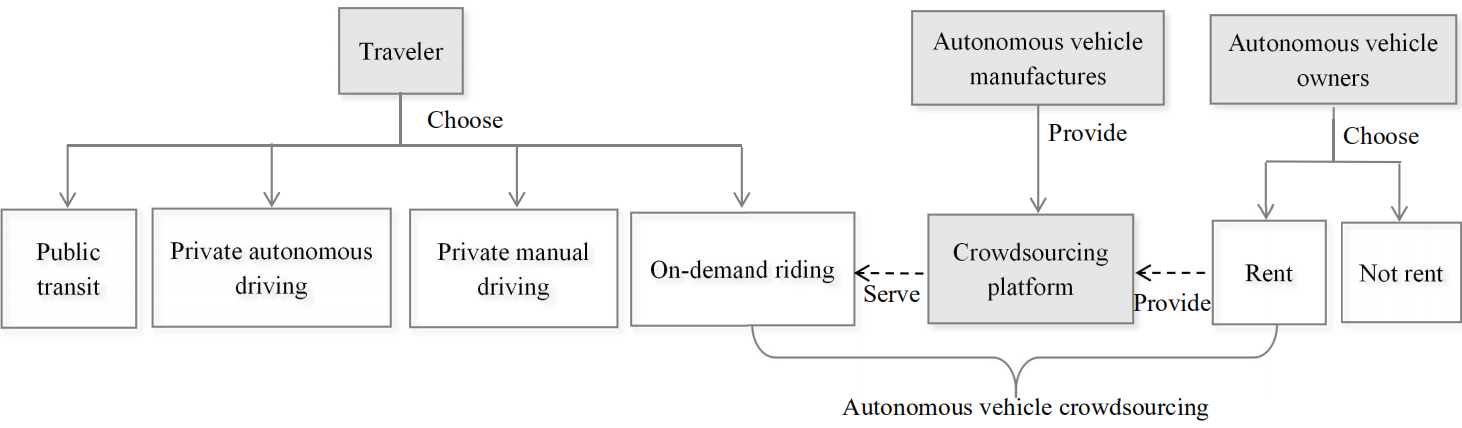}
	\caption[]{Conceptual design of an AV crowdsourcing market} 
	\label{concept_fig}
\end{figure}
\par
\noindent

In the following, we will introduce the details of the equilibrium model establishment. We first describe the equilibrium state of the on-demand ride service; and then, we model the travelers' mode choices as well as the AV owners' rental choices based on the utilities associated with these selections; and finally, we describe the equilibration of supply and demand in a steady-state mobility market with AV renting options.

\subsection{Ride services and traffic congestion} \label{sec3_1}

We separate one-day time into several periods and use piece-wise fixed level of demand to approximate the time-variant demand. In this section, we focus on a single period where the platform set prices to equilibrate the supply and demand of on-demand rides. On-demand ride services are composed of matching, pick-up and on-travel processes. The matching is assumed to be a bilateral process where both the available AVs and the customers actively search for each other. Since this process does not display distinctive differences from the existing on-demand ride services, we can adopt Zha's formula 
to describe the relation between the customers' average searching time $w^c$  and the AVs' average searching time $w^t$. Derived from the Cobb-Douglas matching function, the customers' average searching time function at the equilibrium can be formulated as follows: 

\begin{align}
& w^c= q_o^{(1-\alpha_1-\alpha_2)/\alpha_2} {A^{-1/\alpha_2}} {{(w^t)}^{-\alpha_1/\alpha_2}}\label{eq1}
\end{align}

Upon being matched, the AV drives itself to the customized location to pick up the customer. We assume that the matching algorithm assigns to every request one of the closest $(w^t q_o)$ available vehicles. Given the size of the city $R$, the average pick-up time $ t_p \approx k{(w^t q_o)}^{-1/2} \sqrt{R}/v$, where $k$ is a parameter depending on the network topology and $k\approx0.63$ for networks that resemble rectangular grids and $v$ is the cruising speed. The average trip time from picking up a customer to arriving at the destination is $t_r = \kappa \frac{\sqrt{R}}{v}$, where $\kappa$ is a constant reflecting the average travel distance. Then, we obtain the following equation between the average pick-up time $t_p$ and the average trip time from picking up a customer to the arrival at the destination $t_r$:

\begin{align}
& t_p= \theta \left(w^t q_o\right)^{-\frac{1}{2}}t_r\label{eq2}
\end{align}

\noindent where $\theta = k/\kappa$ is a city-specific constant. \par

Next, we discuss the modeling of traffic congestion. Since we have constructed the relation between the average pick-up time $t_p$ and the average trip time $t_r$, in the following we focus on the effect of traffic congestion on the average trip time $t_r$. The functional form is given by the following equation: 

\begin{align}
& t_r= T_s\left[ q_m + \alpha\left(q_a - q_o\right) + \alpha q_o (1 + \theta (w^t q_o)^{-\frac{1}{2}} ); R, v_0 \right ] \label{eq3}
\end{align}

In the above equation, $T_s(\cdot)$ is a continuously differentiable, monotonically increasing and convex function; $q_m$, $q_a$, $q_o$ are unit-time travel demands for manual driving, autonomous driving (including both private vehicles and crowdsourcing vehicles) and on-demand ride service, respectively; $\alpha \in (0,1)$ is a parameter representing the "relative occupation" of AVs compared to MVs (because the AV technology can reduce the vehicle headway in traffic flows). The equation states that the average trip time $t_r$ is a function related to city size $R$ and free-flow cruising speed $v_0$, and it is a monotonically increasing and convex function with regards to the term $q_m + \alpha(q_a - q_o) + \alpha q_o (1 + \theta (w^t q_o)^{-\frac{1}{2}} )$. In this term, $q_m$ represents the MVs' contribution to congestion, $\alpha (q^a-q_o)$ represents privately-used AVs' contribution to congestion, and $\alpha q_o (1 + \theta (w^t q_o)^{-\frac{1}{2}} )$ represents the crowdsourcing AVs' contribution to traffic congestion; the latter has a factor of $(1 + \theta (w^t q_o)^{-\frac{1}{2}} )$ because it includes both the pick-up trips and delivery trips, where $\theta (w^t q_o)^{-\frac{1}{2}}$ is associated with Eq. (\ref{eq2}). The functional form of $T_s(\cdot)$ can be obtained from realistic data. \par

Despite their analytical formulation, there are some mathematical difficulties associated with Eqs. (\ref{eq1})-(\ref{eq3}). The main difficulties lie in the terms $(w^t)^{-\alpha_1 / \alpha_2}$ in Eq. (\ref{eq1}) and $(w^t q_o)^{-\frac{1}{2}}$ in Eqs. (\ref{eq2}) and (\ref{eq3}), where there are negative powers associated with these variables; this gives rise to discontinuity and possibly results in imaginary numbers in the process of solving the equilibrium. To address this issue, we introduce some continuous counterparts to replace the above discontinuous functions. The idea is as follows (taking $(w^t)^{-\alpha_1 / \alpha_2}$ as example). We first choose a very small positive number $\epsilon$. Then, when $w^{t} \ge \epsilon$, the counterpart takes exactly the same form as $(w^t)^{-\alpha_1 / \alpha_2}$; and when $w^{t} < \epsilon$, the counterpart is a linear function such that the whole function is continuous and smooth at $\epsilon$. The same approximation is used for $(w^t q_o)^{-\frac{1}{2}}$. This approximation is reasonable because $w^t$ and $w^t q_o$ only have physical meaning when they are positive, and a sufficiently small $\epsilon$ can ensure that the approximations are almost the same as the original values on the positive horizon. In the following, we denote the two approximations as $\xi_1(w^t)$ and $\xi_2(w^t q_o)$. The graphical illustration of the approximation is shown in Figure \ref{Approximation_fig}.

\begin{figure}[!ht]
	\centering
	\includegraphics[width=0.6\textwidth]{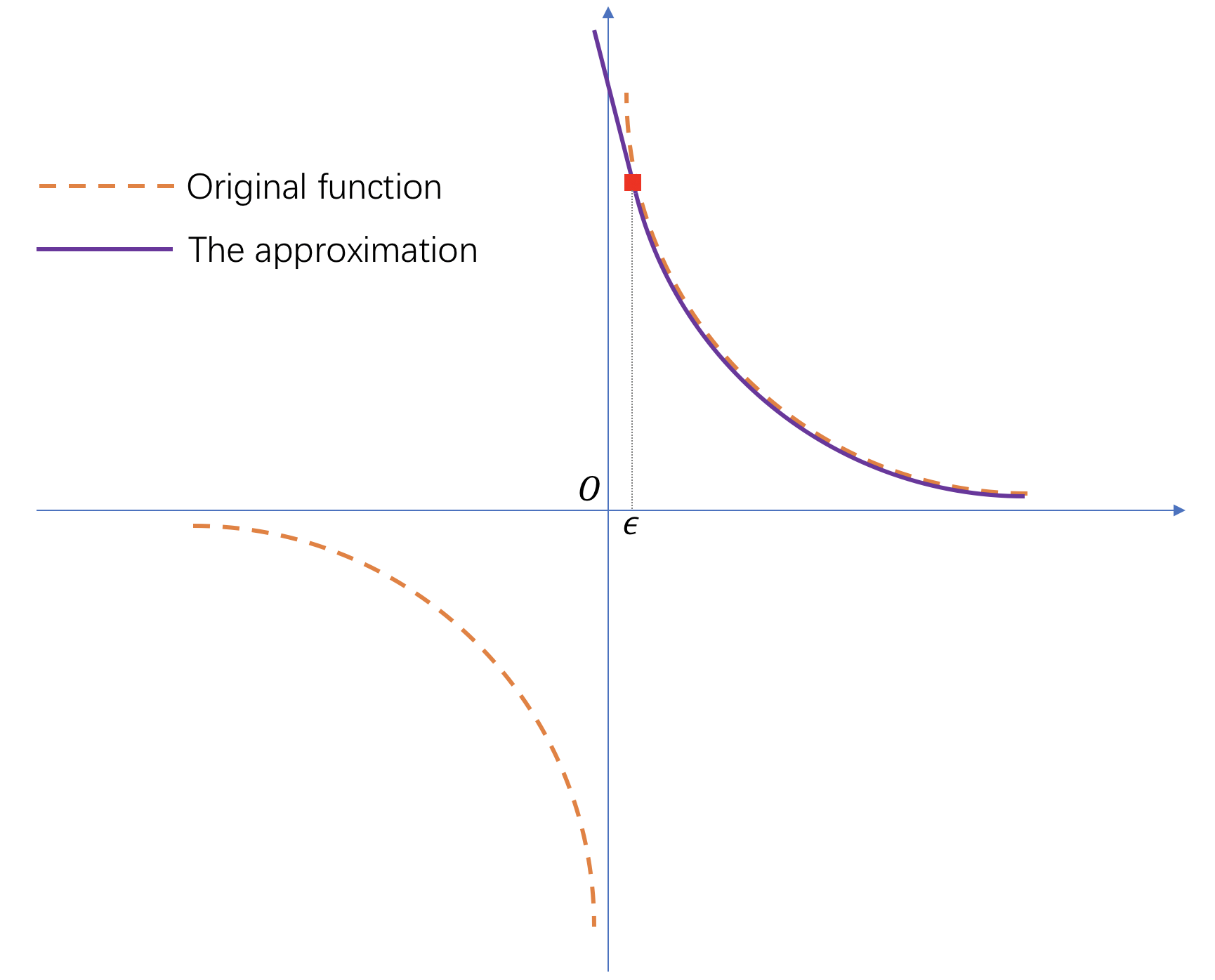}
	\caption[]{Illustration of function approximation} 
	\label{Approximation_fig}
\end{figure}
\par

With the above-mentioned function approximation scheme, we can express the approximated counterparts of Eqs. (\ref{eq1})-(\ref{eq3}) as described below. These expressions will be used in the model analysis and algorithmic development.

\begin{align}
& w^c = W^c (q_o, w^t; A, \alpha_1, \alpha_2) \label{eq1_approx} \\
& t_p = T_p (w^t, q_o, t_r; \theta) \label{eq2_approx} \\
& t_r = T_r (q_m,q_a,q_o,w^t; R, v_0, \theta) \label{eq3_approx}
\end{align}

\subsection{Choice modeling} \label{sec3_2}

\begin{figure}[!ht]
	\centering
	\includegraphics[width=0.65\textwidth]{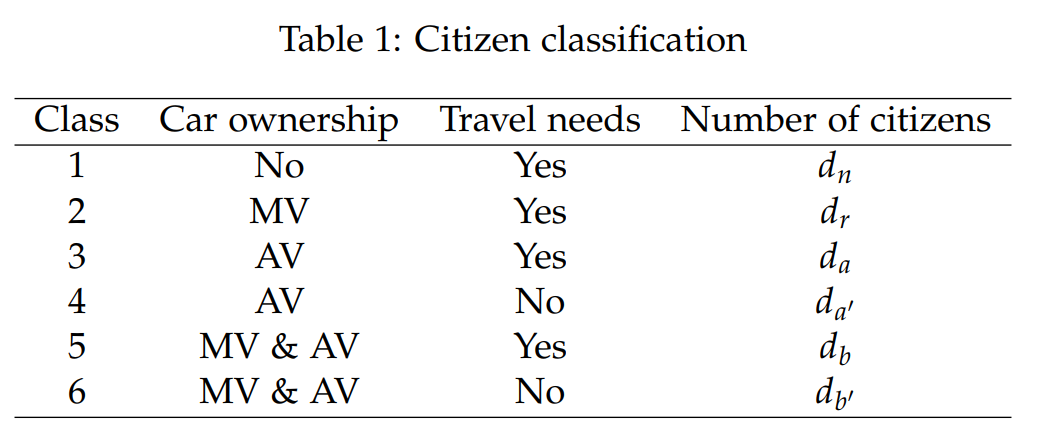}
	\label{class_fig}
\end{figure}
\noindent

AV owners have the incentive to rent their cars to the crowdsourcing platform only if this action can earn profit. Suppose that the rental depends on the serving rates, i.e., for a period with duration $h$, the rental revenue for the AV owners is given by $(pn_0)$, where $p$ is the  payment per ride and $n_0$ is the expected number of on-demand rides served by an AV in this period. The rental revenue can be viewed as the "commission fee" for the vehicle providers. Then, the utility of a citizen renting his/her private AV to the crowdsourcing platform for each period with duration $h$ is given by $(pn_0-m)$, where $m$ is a constant indicating the expected additional cost of sharing the private vehicle for public use, including energy consumption, vehicle depreciation and psychological costs. The AV owners who choose not to rent their cars will receive zero rental utility. 
\par

Next, we discuss the mode choices for citizens with travel needs in this period. The travel (dis)utility includes time costs and monetary costs. We respectively specify the perception of on-board time for AVs, public transit and MVs as $\beta_A$, $\beta_P$ and $\beta_M$, and the relation among them is $\beta_A < \beta_P < \beta_M$; this is because driving a car costs the most human effort, and taking public transit should be generally less comfortable than taking an AV. The perception of waiting for a crowdsourcing AV is specified as $\gamma$ (The waiting time is composed of matching and pick-up time). On the other hand, it is assumed that the average travel time $t_n$ and fare $F_n$ of taking public transit are constant, i.e., unaffected by traffic congestion (e.g., BRT vehicles or subway). Given the average fare $F_o$ for each trip on a crowdsourcing AV, the utility of taking on-demand ride service can be stated as $(-F_o -\beta_A t_r-\gamma(w^c+t_p))$. The utilities of other transport modes can be similarly expressed. The utilities for all choices are shown in Table 2, and the (dis)utility $V_x^{(i,j)}$ can be calculated by adding up the rental utility and travel utility in each choice $(i,j)$ for the citizens of class $x$. 

\begin{figure}[!ht]
	\centering
	\includegraphics[width=1\textwidth]{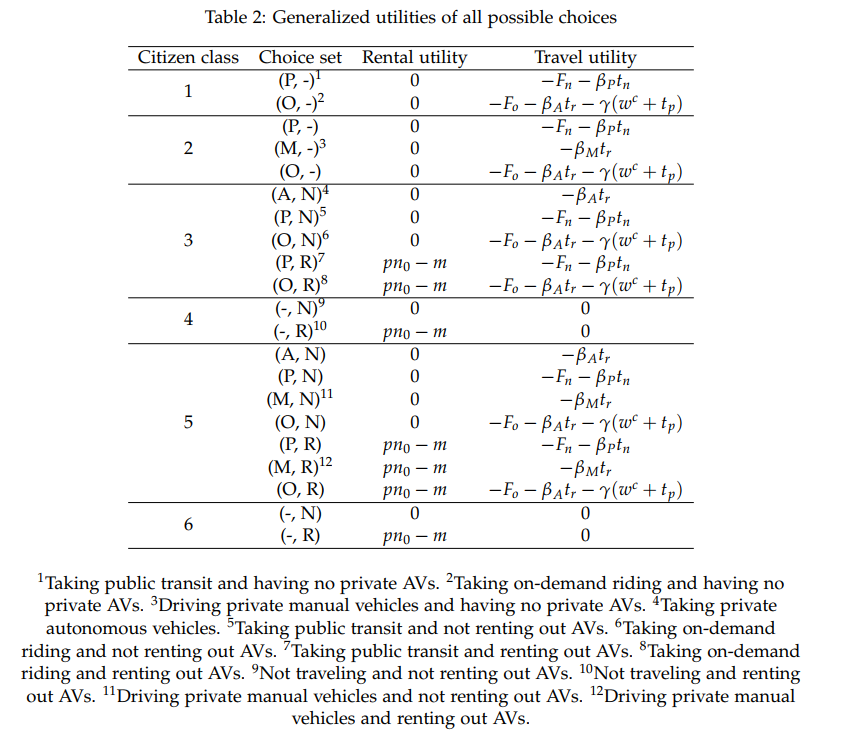}
	\label{utility_fig}
\end{figure}
\noindent

In the AV crowdsourcing market, we assume that sensitivity to travel utility and rental utility for one class of people are identical; and the random variables are identically distributed with a Gumbel density function. Then, the probability for the choice $(i, j)$ by people of class $x$ can be captured by the following logit model: 

\begin{align}
& \pi_x^{(i,j)}=\frac{e^{\mu_x V_x^{(i,j)}}}{\sum\nolimits_{(i,j) \in \mathcal{C}_x}^{} e^{\mu_x V_x^{(i,j)}}} & \forall (i,j) \in \mathcal{C}_x, x \in \{1,2,\dots,6\}\label{eq4}
\end{align}

\noindent
where $\mu_x$ is a nonnegative parameter representing the degree of uncertainty for the choices from the perspective of people of class $x$. A larger parameter $\mu_x$ means that people have more knowledge about the AV crowdsourcing market and thus can be more certain about their utilities, i.e., an increasing number of class $x$ citizens would give up the choice $(i,j)$ when $V_x^{(i,j)}$ decreases by a unit. In modeling customer choices, however, it is currently unknown that whether the sensitivity coefficients associated with travel and rental choices are similar. If they are not identical, we need a nested logit model to capture the choice probabilities. The detailed discussions of the nested logit models are presented in Appendix B.

\subsection{Supply and demand} \label{sec3_3}

This section establishes the equilibration between supply and demand. For crowdsourcing AVs, the supply equation can be written as:

\begin{align}
& N=N_s+N_r\label{eq5}
\end{align}

\noindent
where $N$ is the total number of usable AVs for the crowdsourcing platform in this period; $N_s$ is the number of pre-purchased AVs by the crowdsourcing platform; and $N_r$ is the number of AVs collected from crowdsourcing in this period. The number of AVs crowdsourced from the society can be expressed by the following equation:

\begin{align}
& N_r=d_a^{}\left(\pi_a^{P,R}+\pi_a^{O,R}\right)+d_{a'}^{}\pi_{a'}^{-,R}+d_b\left(\pi_b^{P,R}+\pi_b^{M,R}+\pi_b^{O,R}\right)+d_{b'}^{}\pi_{b'}^{-,R}\label{eq6}
\end{align}
\noindent

\noindent where the first to last terms of the RHS represent the number of AVs from citizens belonging to classes 3-6, respectively.\par

Similarly, on the demand side, the equilibration can be expressed as:

\begin{align}
& q_o h=d_n^{}\pi_n^{O,-}+d_r^{}\pi_r^{O,-}+d_a^{}\left(\pi_a^{O,N}+\pi_a^{O,R}\right)+d_b\left(\pi_b^{O,N}+\pi_b^{O,R}\right)\label{eq7}\\
& q_a h=q_o h + d_a\pi_a^{A,N}+d_b\pi_b^{A,N}\label{eq8}\\
& q_m h=d_r\pi_r^{M,-}+d_b\left(\pi_b^{M,N}+\pi_b^{M,R}\right)\label{eq9}\\
& q_p h=d_n\pi_n^{P,-}+d_r\pi_r^{P,-}+d_a\left(\pi_a^{P,N} + \pi_a^{P,R}\right)+d_b\left(\pi_b^{P,N}+\pi_b^{P,R}\right)\label{eq10}
\end{align}
\noindent

In Eq. (\ref{eq7}), the LHS is the total number of customers using the on-demand ride services in this period, and the RHS is the sum of the travelers choosing the associated mode. Eqs. (\ref{eq8})-(\ref{eq10}) express similar equilibria for the AV riders (including privately used ones and crowdsourcing ones), for manually driving, and for taking public transit, respectively.\par

According to the stability requirement, the number of trips of the on-demand rides should be equal to the total number of crowdsourcing AV riders during a period with duration $h$:

\begin{align}
& Nn_0=q_o h \label{eq11}
\end{align}
where $N$ is the total number of crowdsourcing AVs and $n_0$ is the expected number of the on-demand rides served by an AV in this period. 
\noindent

Lastly, the sum of searching time, pick-up time and on-board time for all AVs put into use should be equal to their total service time. Therefore, the service time constraint in view of unit time is modeled as follows:

\begin{align}
& N=q_o \cdot \left(w^t+t_p+t_r\right)\label{eq12}
\end{align}
where $q_o$ is the number of crowdsourcing trips and $\left(w^t+t_p+t_r\right)$ is the service time for a single trip by a crowdsourcing AV.
\noindent

Aggregating Eqs. (\ref{eq1_approx})-(\ref{eq12}) yields the equilibrium model of the AV crowdsourcing market. The following subsection will specify some basic properties of the model.

\subsection{Existence of equilibrium solution} \label{sec3_4}

For the platform, on the demand side, the trip fare of on-demand ride service $F_o$ is treated as the decision variable; on the supply side, the commission fee per ride $p$ is treated as the decision variable. When fixing these two decision variables, the utility equations along with Eqs. (\ref{eq1_approx})-(\ref{eq12}) together constitute a system of nonlinear equations that can be used to obtain the equilibrium state of the AV crowdsourcing market, i.e., the intermediate variables including choice splits of different travel modes, on-demand ride service rate and total travel costs. \par

We examine the existence of the equilibrium solution in this subsection. The statement is presented as follows.

\begin{prop}
The system of nonlinear equations (\ref{eq1_approx})-(\ref{eq12}) has at least one solution when fixing $F_o$ and $p$.
\end{prop}

\begin{proof}

We apply the fixed point theorem to prove the existence. The proof is based on Schauder's fixed-point theorem: let $\mathcal{C}$ be a closed convex subset of the Banach space and suppose $f:\mathcal{C} \mapsto \mathcal{C}$ and $f$ are compact (i.e., bounded sets in $\mathcal{C}$ are mapped into relatively compact sets), and then $f$ has a fixed point in $\mathcal{C}$. \par

Let the constant parameters $d_{mv}$ and $d_{av}$ respectively denote the total number of MVs and AVs owned by the citizens. We focus on four variables $(t_r, n_0, t_p, w^c)$, and the system (\ref{eq1_approx})-(\ref{eq12}) can be treated as a mapping $\Phi$ such that $(\hat{t}_r, \hat{n}_0, \hat{t}_p, \hat{w}^c) = \Phi (t_r, n_0, t_p, w^c)$, and the equilibrium requires $(\hat{t}_r, \hat{n}_0, \hat{t}_p, \hat{w}^c) = (t_r, n_0, t_p, w^c)$. The mapping works as follows: it first computes all choice splits with the logit model Eq. (\ref{eq4}), and obtain $\hat{n}_0$ by Eq. (\ref{eq11}), obtain $w^t$ by Eq. (\ref{eq12}); and then with $w^t$ we obtain $\hat{t}_r$, $\hat{t}_p$ and $\hat{w}^c$ by Eqs. (\ref{eq1_approx})-(\ref{eq3_approx}). Clearly, the mapping $\Phi$ is continuous. To proceed, we first prove the following result. \par

\textbf{Claim}: \underline{when $\max(t_r, t_p) \rightarrow +\infty$, $\hat{t}_r, \hat{t}_p$ are upper bounded by constants.}\par

To prove this claim, we note a simple fact that when $t_r$ or $t_p$ approaches infinity, $q_o$ will approach zero with an exponential rate due to the presence of constant-utility public transit, so that $q_o(t_r + t_p)$ will also approach zero. Therefore, by Eq. (\ref{eq12}), $q_o w^t \approx N$. Meanwhile, $N$ is lower bounded by a positive constant due to the presence of citizens with types 4 and 6 and the fact that $n_0 \ge 0$. Thus, $q_o w^t$ is lower bounded by a constant. Meanwhile, $q_m$ and $q_a$ are upper bounded by $d_{mv}$ and $d_{av}$, respectively. Combining the above facts, based on Eqs. (\ref{eq2_approx}) and (\ref{eq3_approx}), we know that $\hat{t}_r, \hat{t}_p$ are upper bounded by constants. The claim is proved.\par

On the other hand, $\hat{n}_0$ is always upper bounded by a constant because $N$ is lower bounded by a positive constant, and with Eq. (\ref{eq11}), we can easily identify the bounded nature of $\hat{n}_0$. \par

Now, we define the ranges of $t_r$, $t_p$ and $n_0$ to be within $[0,M_1]$, where $M_1$ is a large positive number. We show that $\hat{w}^c$ is upper bounded by a constant; with Eq. (\ref{eq1_approx}), this is equivalent to showing that $w^t$ is lower bounded by a constant. By Eq. (\ref{eq12}), we obtain:

\begin{align*}
& w^t = \frac{N}{q_o} - t_p - t_r \ge - t_p - t_r \ge - 2M_1
\end{align*}

\noindent and with which the upper boundedness of $\hat{w}^c$ is confirmed.\par

Finally, we define the feasible ranges of $t_r$, $t_p$, $w^c$ and $n_0$ to be all within $[0,M_1]$, and this region is denoted as $\mathcal{D}$. When $M_1$ is sufficiently large, by the above reasoning we know that $\hat{w}^c$ and $\hat{n}_0$ are upper bounded by constants smaller than $M_1$, and based on the claim stated above as well as the continuity of $\Phi$, we know that $\hat{t}_r, \hat{t}_p$ are also upper bounded by constants smaller than $M_1$. Therefore, $\Phi (t_r, n_0, t_p, w^c) \in \mathcal{D}$, and by Schauder's fixed-point theorem the existence is guaranteed. The proof is completed. 

\end{proof}

\section{Scenario Analysis} \label{sec4}

Using the equilibrium model proposed in the previous section, we can then proceed to analyze several scenarios with different objectives for decisions regarding the platform payment $p$ and the platform charge $F_o$. The scenarios of interest include a monopoly scenario, a first-best scenario and a second-best scenario. The first intends to maximize the crowdsourcing platform’s profit; the second intends to maximize social welfare; and the third intends to maximize social welfare under the constraints that the crowdsourcing platform's profit is guaranteed to be no less than a preset threshold. For practical consideration, in this section we assume that the time-of-day can be divided into several heterogeneous periods according to the levels of demands. In the following, we use $\mathcal{P}$ to denote the set of all homogeneous periods, and $H_k$ represents the duration of the $k^{th}$ period, $\forall k \in \mathcal{P}$. Below, we first present the optimization models in the above-mentioned three scenarios, and then develop the algorithms for solving these models.\par

In this section, we use a tuple $\boldsymbol{\tau}$ to collectively denote all of the equilibrium variables defined in the last section; these are treated as intermediate variables in the established optimization models.

\subsection{Monopoly scenario} \label{sec4_1}

In the monopoly scenario, the crowdsourcing platform decides some key variables to maximize its daily profit. The decision variables include the period-specific average fare for the on-demand ride service $F_{o,k}$ and the period-specific payment to AV owners per customer $p_k$. The total number of purchased AVs, i.e., $N_{s}$, is a parameter that is not considered as a decision variable. Let $q_{o,k}$ denote the period-specific unit-time customer demand for the on-demand ride service; $N_{r,k}$ denote the period-specific number of rented AVs; and $n_{0,k}$ denote the period-specific number of customers served by a crowdsourcing AV per period of duration $h$. The optimization model for the profit maximization problem is then stated as: 

\begin{align}
& \max_{\mathbf{F}_o, \mathbf{p}, \boldsymbol{\tau}}  \ \left\{ \sum_{k \in \mathcal{P}}\frac{H_k}{h}\left(F_{o,k} q_{o,k} h - N_{r,k}p_k n_{0,k}\right) \right\} -N_{s}\left(g+z\right)-C_f \label{eq13} \\
& \mathrm{s.t.} \ (\ref{eq1_approx})-(\ref{eq12}) & \forall k \in \mathcal{P} \nonumber \\
& F_{o,k}, p_k \ge 0 & \forall k \in \mathcal{P} \label{eq13_1} 
\end{align}

\noindent
where $g$ and $z$ are the unit-day amortized purchase cost and maintenance cost of an AV respectively; and $C_f$ represents the basic operational cost of the crowdsourcing platform per day.

In the objective function (\ref{eq13}), $H_k F_{o,k} q_{o,k}$ represents the revenue of the platform in the $k^{th}$ period, and $\frac{H_k}{h} N_{r,k}p_k n_{0,k}$ is the total payment to the AV owners in the $k^{th}$ period. Eq. (\ref{eq13_1}) states that the fares and payments should be nonnegative. It should be noted that the constraints (\ref{eq1_approx})-(\ref{eq12}) are period-specific; i.e., for each period $k$, there is a corresponding constraint set stating the equilibration of the AV crowdsourcing market.

\subsection{First-best scenario} \label{sec4_1}

Now we consider the case that a government agency is operating the crowdsourcing platform, and its goal is to maximize social welfare. This scenario is called the first-best scenario. By deciding the period-specific fare $F_{o,k}$ and payment $p_k$, the optimization model for the operating agency can be written as:

\begin{align}
& \max_{\mathbf{F}_o, \mathbf{p},\boldsymbol{\tau}}  \sum_{k \in \mathcal{P}} \frac{H_k}{h}S_k \label{eq14} \\
& \mathrm{s.t.} \ (\ref{eq1_approx})-(\ref{eq12}), (\ref{eq13_1}) & \forall k \in \mathcal{P} \nonumber
\end{align}

\noindent
where $S_k$ represents the unit-time social welfare in the $k^{th}$ period, and the length of such time duration is $h$. The constraints are the same as in the monopoly scenario. The social welfare is composed of the total time cost by various transport modes and an additional cost of sharing private AVs with other people; its detailed form is given by:

\begin{align}
& S_k = \sum_{x \in \{1,2,\dots,6\}} \frac{1}{\mu_{x}} d_{x,k} \left( \ln \sum_{(i,j) \in \mathcal{C}_x} e^{\mu_x V_{x,k}^{(i,j)}} \right)+\frac{H_k}{h}\left(F_{o,k} q_{o,k} h - N_{r,k}p_k n_{0,k}\right)& k \in \mathcal{P}\label{eq15}
\end{align}
In Eq. (\ref{eq15}),  $\frac{1}{\mu_{x}} d_{x,k} ( \ln \sum_{(i,j) \in \mathcal{C}_x} e^{\mu_x V_{x,k}^{(i,j)}})$ is the total utility of class $x$ citizens, where $\mu_x$ is the corresponding scale parameter; $d_{x,k}$ is the period-specific population of class $x$ citizens; and $V_{x,k}^{(i,j)}$ is the period-specific (dis)utility of choice $(i,j)$ by class $x$ citizens. The social welfare should exclude monetary utility, and therefore we compensate by adding $\frac{H_k}{h}\left(F_{o,k} q_{o,k} h - N_{r,k}p_k n_{0,k}\right)$ in calculating the welfare terms.

\subsection{Second-best scenario} \label{sec4_3}

From the platform's perspective, the on-demand ride service can produce a deficit under the first-best case, and a similar phenomenon is identified in the traditional taxi and ride-sourcing markets. Specifically, in the early stage of AV adoption, it is likely that only a few citizens own AVs, and there will be a shortage of AVs that can be rented. In this case, if the crowdsourcing platform lacks start-up fund to purchase AVs, the on-demand ride service provided by the crowdsourcing AVs will be of little attraction since it requires excessively long pick-up distances, cutting down the platform's profit. To guarantee a certain level of profitability for the service provider, in this subsection we propose a second-best operational strategy that maximizes social welfare under the constraint of a lower bound on the platform profit. The model formulation is illustrated as follows:

\begin{align}
& \max_{\mathbf{F}_o, \mathbf{p},\boldsymbol{\tau}}  \sum_{k \in \mathcal{P}} \frac{H_k}{h}S_k \label{eq16} \\
& \mathrm{s.t.} \ (\ref{eq1_approx})-(\ref{eq12}), (\ref{eq13_1}) & \forall k \in \mathcal{P} \nonumber \\
&  \sum_{k \in \mathcal{P}}\frac{H_k}{h}\left(F_{o,k} q_{o,k} h - N_{r,k}p_k n_{0,k}\right)\ge \rho \left[N_{s}\left(g+z\right) + C_f\right]  & \forall k \in \mathcal{P} \label{eq17}
\end{align}

\noindent where $\rho \ge 0$ represents the basic rate for revenue accumulation. The second-best formulation shows some differences with the previous two models in that it contains an additional nonlinear constraint with inequalities, i.e., Eq. (\ref{eq17}). The next subsection will show how to deal with this constraint.
\noindent

\subsection{Solution algorithms} \label{sec4_4}

The optimization models presented above are all nonlinear programs with complicated nonlinear constraints. To efficiently solve these models, we treat them as bi-level forms containing an equilibrium problem in the lower level, and then incorporate them into a gradient projection algorithmic framework. In detail, we first obtain the equilibrium given the fare $F_{o,k}$ and payment $p_k$; and then we calculate the gradient of the objective functions to $F_{o,k}$ and $p_k$ in the current equilibrium state, according to which we update the decision variables and project them to the feasible domain when they are to cross the boundary. Thus, the process is carried out iteratively until the termination criterion is met (Figure \ref{algorithms_fig}).
\begin{figure}[!ht]
	\centering
	\includegraphics[width=0.6\textwidth]{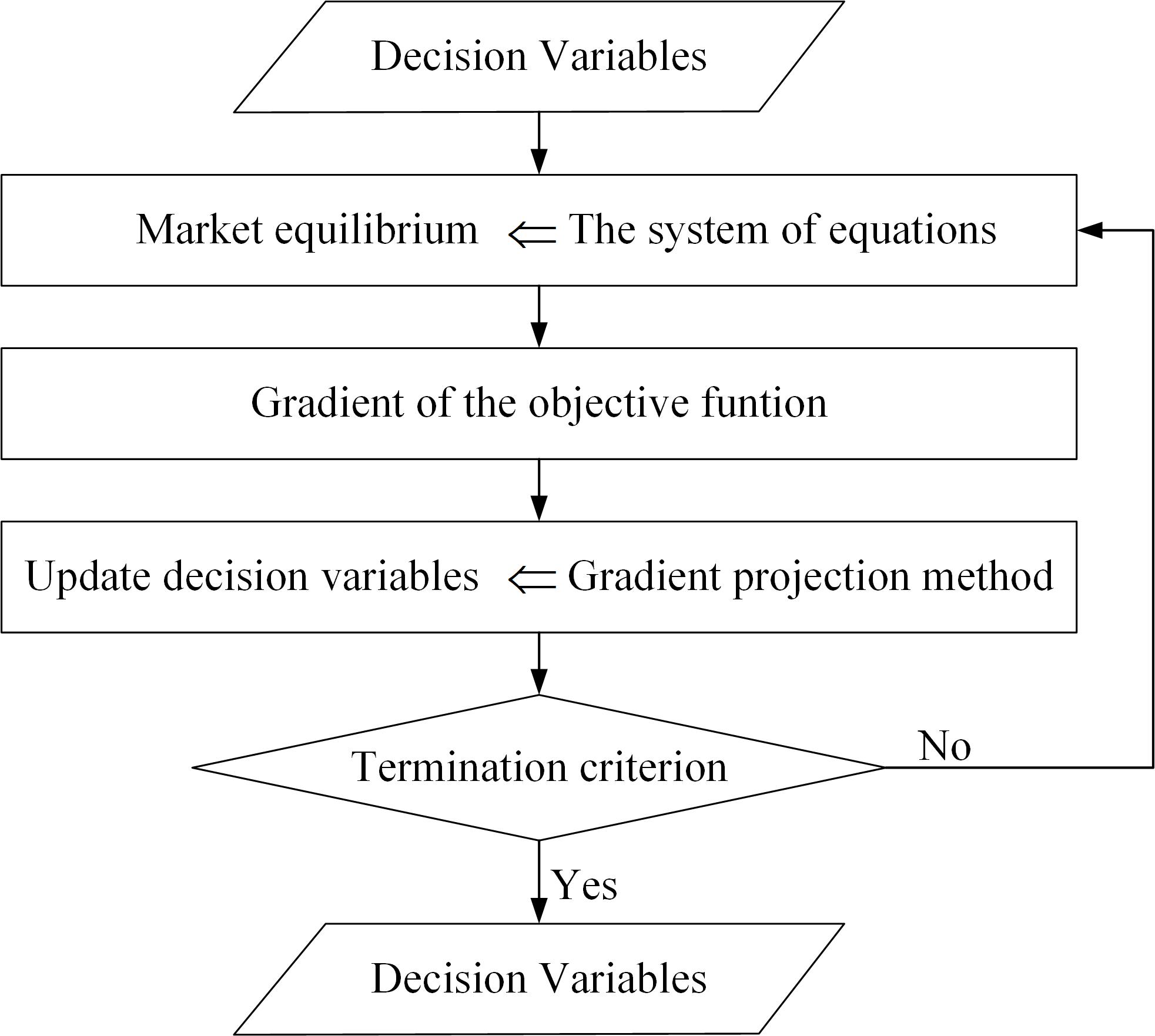}
	\caption[]{Solution algorithms} 
	\label{algorithms_fig}
\end{figure}
\par

We use the implicit function theorem to compute the gradient of equilibrium state $\nabla \mathbf{f}$  with regard to $F_{o,k}$ and payment $p_k$. Applying the chain rule, we have:

\begin{align}
& J(\frac{\partial{\boldsymbol{\tau}}}{\partial{\mathbf{F}_o}},\frac{\partial{\boldsymbol{\tau}}}{\partial{\mathbf{p}}})=-J^T(\frac{\partial{\mathbf{f}}}{\partial{\mathbf{F}_o}},\frac{\partial\mathbf{f}}{\partial{\mathbf{p}}})\cdot J^{-1}(\frac{\mathbf{\partial{f}}}{\partial{\boldsymbol{\tau}}})  
\end{align}

Here, $\mathbf{f}_u$ are the choice utility functions; $\mathbf{f}_\pi$ are the choice probability functions for which the original forms are given by Eq. (\ref{eq4}); $\mathbf{f}_n$ are the supply/demand functions for which the original forms are given by Eqs. (\ref{eq5})-(\ref{eq12}); and $\mathbf{f}_t$ are the traffic-related functions for which the original forms are given by Eqs. (\ref{eq1_approx})-(\ref{eq3_approx}). Then, the equilibrium can be compactly expressed as $\mathbf{f}=\left( \begin{matrix} \mathbf{f}_u \\ \mathbf{f}_\pi \\ \mathbf{f}_n \\ \mathbf{f}_t \end{matrix} \right) = \boldsymbol{0}$. 

The Jacobian matrix $J(\frac{\mathbf{\partial{f}}}{\partial{\mathbf{F}_o}},\frac{\mathbf{\partial{f}}}{\partial{\mathbf{p}}})$ can be obtained as follows. With the exception of $\mathbf{f}_u$ that is linearly related with $\mathbf{F}_o$ and $\mathbf{p}$, all other functions in $\mathbf{f}$ have zero partial derivatives with regard to $\mathbf{F}_o$ and $\mathbf{p}$. To obtain the Jacobian matrix $J(\frac{\mathbf{\partial{f}}}{\partial{\boldsymbol{\tau}}})$, we observe that $\mathbf{f}_u$ and $\mathbf{f}_n$ are linearly related with $\boldsymbol{\tau}$. For probability functions, $\mathbf{f}_\pi$, we have
\begin{align}
& \frac{\partial{f_{\pi,x}^{(m,n)}}}{\partial{V_{x}^{(m,n)}}}=\mu_x \  e^{\mu_x V_{x}^{(m,n)}}\ \frac{ \displaystyle\sum_{(i,j)\in \mathcal{H}_{x}} {e^{\mu_x V_{x}^{(i,j)}} }-{e^{\mu_x V_x^{(m,n)}}} }{ \left(\displaystyle\sum_{(i,j)\in \mathcal{H}_{x}}e^{ \mu_x V_{x}^{(i,j)}}\right)^2} \\
& \frac{\partial{f_{\pi,x}^{(m,n)}}}{\partial{V_{x}^{(w,z)}}}=-\mu_x \  e^{\mu_x V_{x}^{(m,n)}}\frac{ e^{\mu_x V_{x}^{(w,z)}} }{ \left(\displaystyle\sum_{(i,j)\in \mathcal{H}_{x}}e^{\mu_x V_{x}^{(i,j)}}\right)^2}
\end{align}
where $(w,z)\in \mathcal{C}_{x}$, $(m,n)\in \mathcal{C}_{x}$ and $(w,z) \ne (m,n)$. The traffic-related functions $\mathbf{f}_t$ are piece-wise functions consisting of a linear function and a power function. $\frac{\partial\mathbf{f}_t}{\partial\boldsymbol{\tau}}$ can be obtained using the derivative rule for complex functions. 
\noindent

In the second-best scenario, there exists a nonlinear constraint Eq. (\ref{eq17}) that can make the gradient projection method proposed above ineffective. To handle this issue, we introduce a Lagrange multiplier $\lambda$ to transfer the constraint into a term in the objective function. Specifically, the altered objective function can be written as: 

\begin{align}
&\max_{\mathbf{F}_o, \mathbf{p},\boldsymbol{\tau}}  \sum_{k \in \mathcal{P}} \frac{H_k}{h}S_k + \lambda \left\{\left[ \sum_{k \in \mathcal{P}}\frac{H_k}{h}(F_{o,k} q_{o,k} h - N_{r,k}p_k n_{0,k}) \right] -\rho \left[N_{s}\left(g+z\right)+C_f\right] \right\}\label{eq18}
\end{align}

Since the Lagrange multiplier $\lambda$ is unknown prior to solving the problem, in the solution process we must adjust its value accordingly in order to obtain a maximum objective function value while ensuring the feasibility of the constraint. In most cases, the optimal solution is identified at the points where $\sum_{k \in \mathcal{P}}\frac{H_k}{h}\left(F_{o,k} q_{o,k} h - N_{r,k}p_k n_{0,k}\right) = \rho\left[N_{s}\left(g+z\right)+C_f\right]$.

\section{Numerical Examples}\label{sec5}

We report the equilibrium results under different scenarios in this section. 

\subsection{Test settings}\label{sec5_1}

We consider a city with the area of approximately 400 $\text{km}^2$. The population density of this city is $5,000$ persons per square kilometer, and $30\%$ of the residents of the city own MVs or AVs. We set the AV ownership as approximately $3\%$; i.e., the number of AVs owned by citizens is equal to $3\%$ of the population. The time of a day can be divided into two periods according to travel demand: peak and off-peak hours that last $4 \ \mathrm{h}$ and $20 \ \mathrm{h}$, respectively. In the peak-hours, the trip generation per hour equals $15\%$ of total population; in the off-peak hours, the trip generation per hour equals $5\%$ of total population. The model input data for the population of six types of citizens are shown in Table 3.

\begin{figure}[!ht]
	\centering
	\includegraphics[width=0.7\textwidth]{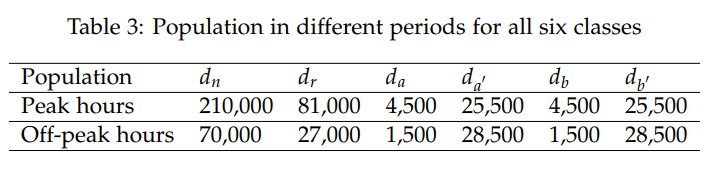}
	\label{test_fig}
\end{figure}
\noindent

AV technological maturity index $\alpha$ is set to $0.7$. The matching function of the AV crowdsourcing is captured by $A=1.5$ and $\alpha_1=\alpha_2=0.7$. The decision period is $1$ hour. An AV purchased by the crowdsourcing platform costs $\$100,000$ with 10 years’ service life and its maintenance fee is $\$5,000$ per year. The pre-purchased number of AVs by the crowdsourcing platform is $N_s=0$; i.e., all supplied AVs are crowdsourced from the AV owners. The cost incurred by the crowdsourcing platform is set to $\$ 6\times10^5$ per day. The additional cost of AV owners is set as $m=\$ 20$. Public transit is captured by the average travel time $t_n=1 \ \mathrm{h}$ and the average fare $F_n=\$6$. For simplicity, we consider $\mu_x=\mu$ for all classes of citizens. Traveler’s VOTs are $\beta_A=\$ 20/\text{h}$, $\beta_P=\$ 30/\text{h}$, $\beta_M=\$40/\text{h}$,$\gamma=\$ 30/\text{h}$. 
\noindent

Setting $\kappa=1$, the average trip time with free-flow speed by private vehicle is $30 \ \mathrm{min}$. Without loss of generality, the trip time function considering congestion $T_s(\cdot)$ is set as a quadratic function, i.e., $T_s(\cdot)=a+b[ q_m + \alpha(q_a - q_o) + \alpha q_o (1 + \theta (w^t q_o)^{-1/2} )]^2 $, and the approximation takes similar form. The basic trip time term is $a=0.5 \ \mathrm{h}$, and the congestion related parameter is $b=2.67\times10^{-12} \ \mathrm{h}$, indicating that traffic congestion can lead to a delay of at most $1 \ \mathrm{h}$. 
\subsection{Results} \label{sec5_2}
The first-best scenario reduces people's travel cost and increases the AV owners' revenue significantly (Table 4) illustrating that the first-best scenario encourages AV owners to share the private AVs with others and encourages citizens to travel by on-demand ride service. Payment in the first-best scenario is high: considering an AV owner renting out their private AV in off-peak hours every day, it requires approximately 1.12 years for the owner to recoup the purchase cost of the vehicle. This suggests that AV crowdsourcing may be a worthwhile investment for some citizens; in turn, this will promote AV acceptance by the public. Different from the current manual driving on-demand ride providers, we find the AV crowdsourcing platform can still greatly benefit from the first-best scenarios (Table 4). However, the results show that the crowdsourcing platform must cut approximately $62\%$ of its profits to achieve the first-best scenario. To balance the crowdsourcing platform's profits and the overall social welfare, we explore the second-best scenario and find a trade-off pricing strategy in which the social welfare is close to that of the first-best scenario while the crowdsourcing platform only loses approximately $35\%$ of its monopoly profits. By applying this second-best pricing strategy, it is easier to obtain the agreement of the crowdsourcing platform to  achieve a relatively high social welfare. We note that $\mu=0.1$ is set in the tests.
\begin{figure}[!ht]
	\centering
	\includegraphics[width=1\textwidth]{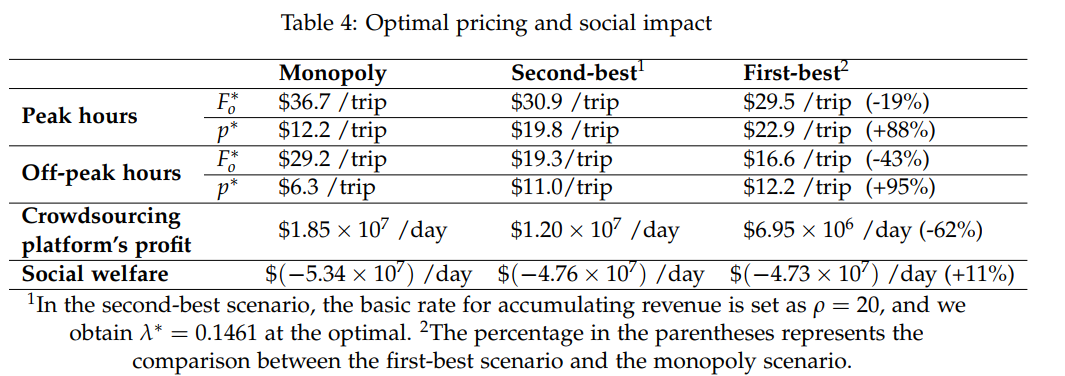}
	\label{test_fig}
\end{figure}

\begin{figure}[!ht]
	\centering
	\includegraphics[width=0.8\textwidth]{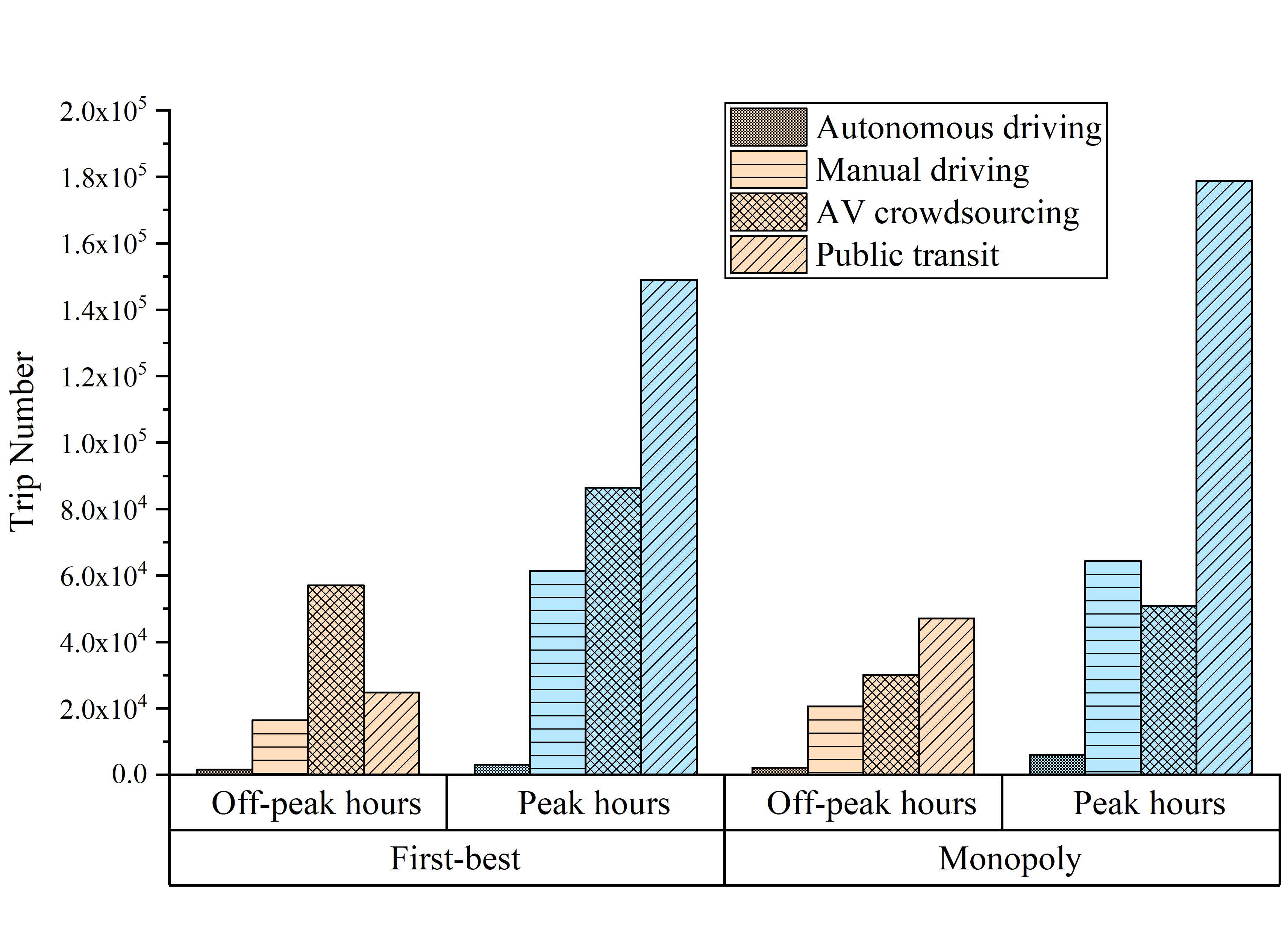}\\
	\caption[]{Trip number per day of different modes.} 
	\label{fig_modesplit}
\end{figure}
\par
\noindent

To examine the mode split in the equilibrium, one notable finding is that in the first-best scenario, the number of AV crowdsourcing services is increased by roughly $90\%$ during peak hours and by $70\%$ in the off-peak hours compared with the monopoly scenario (Figure \ref{fig_modesplit}) as a result of higher payment and lower fare. This illustrates that AV crowdsourcing service has an overall beneficial effect on the society, outweighing the congestion externality caused by deadheading. Additionally, in the first-best scenario, roughly $87.8\%$ and $76.0\%$ of the AV owners in the peak and off-peak hours who have no travel needs choose to rent AVs to the crowdsourcing platform, respectively; roughly $57.2\%$ and $42.3\% $ of the AV owners in peak and off-peak hours who have travel needs choose to rent AVs to the crowdsourcing platform, respectively. We find the payment so attractive that most people choose to rent their AVs to the crowdsourcing platform rather than drive them themselves. Other findings include that AV crowdsourcing plays a more important role in the off-peak hours, i.e., serving up to $45\%$ of the total trips for the first-best scenario; while in other cases, public transit is the most utilized transport mode. It is important to note that in the first-best scenario, people will shift to other modes from public transit due to its high time cost.

\subsection{Sensitivity analysis} \label{sec5_3}

In this section, we examine the influence of a series of model parameters on the equilibrium results, including people's recognition of utility uncertainties (i.e., the scale parameter in the logit model), population density, AV market penetration rate, AV technology maturity and the additional cost for sharing one's AVs with the society (i.e., $m$). 

\subsubsection{Scale parameter in the logit model}
\begin{figure}[h]
	\centering
	\subfloat[][]{\includegraphics[width=0.5\textwidth]{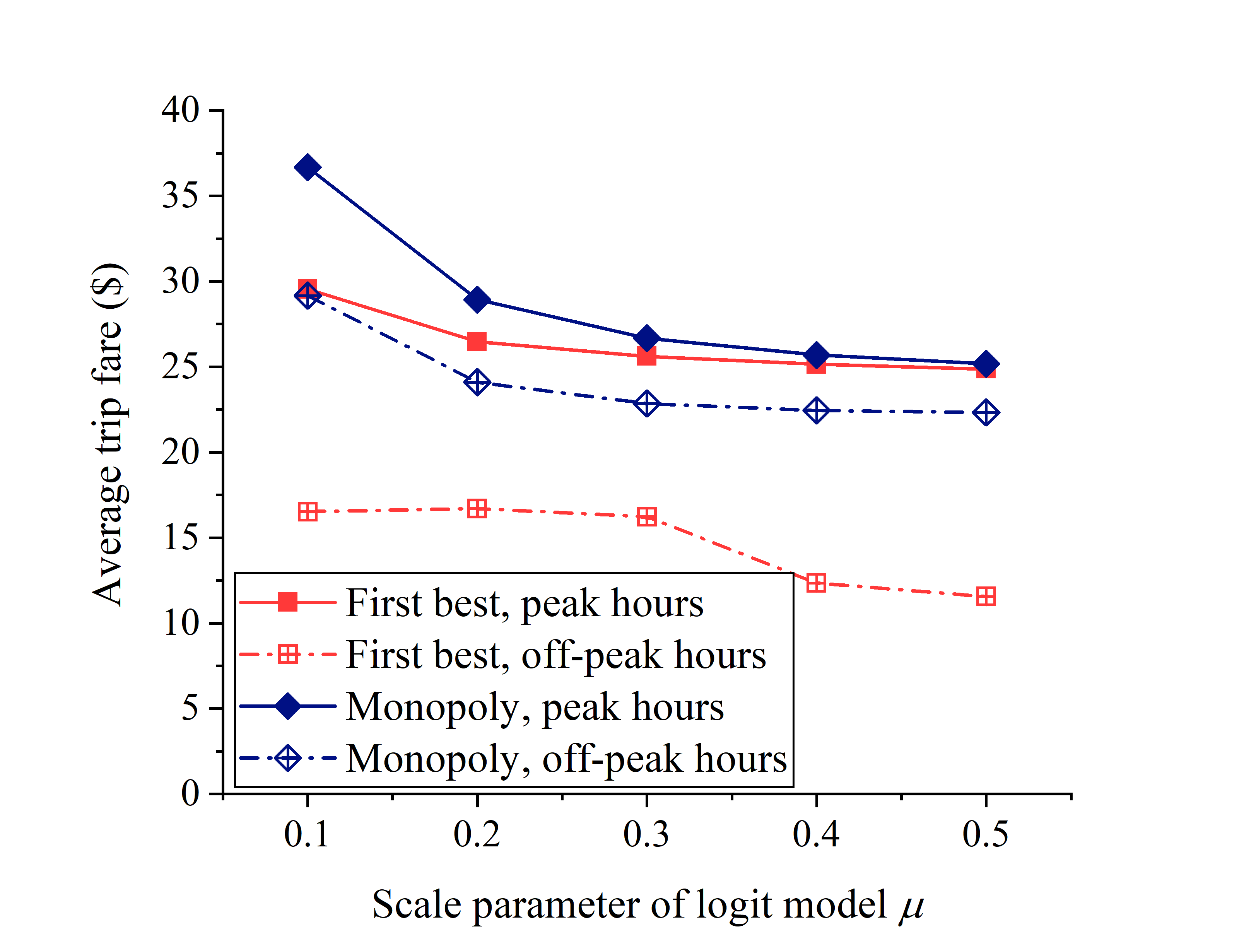}}
	\subfloat[][]{\includegraphics[width=0.5\textwidth]{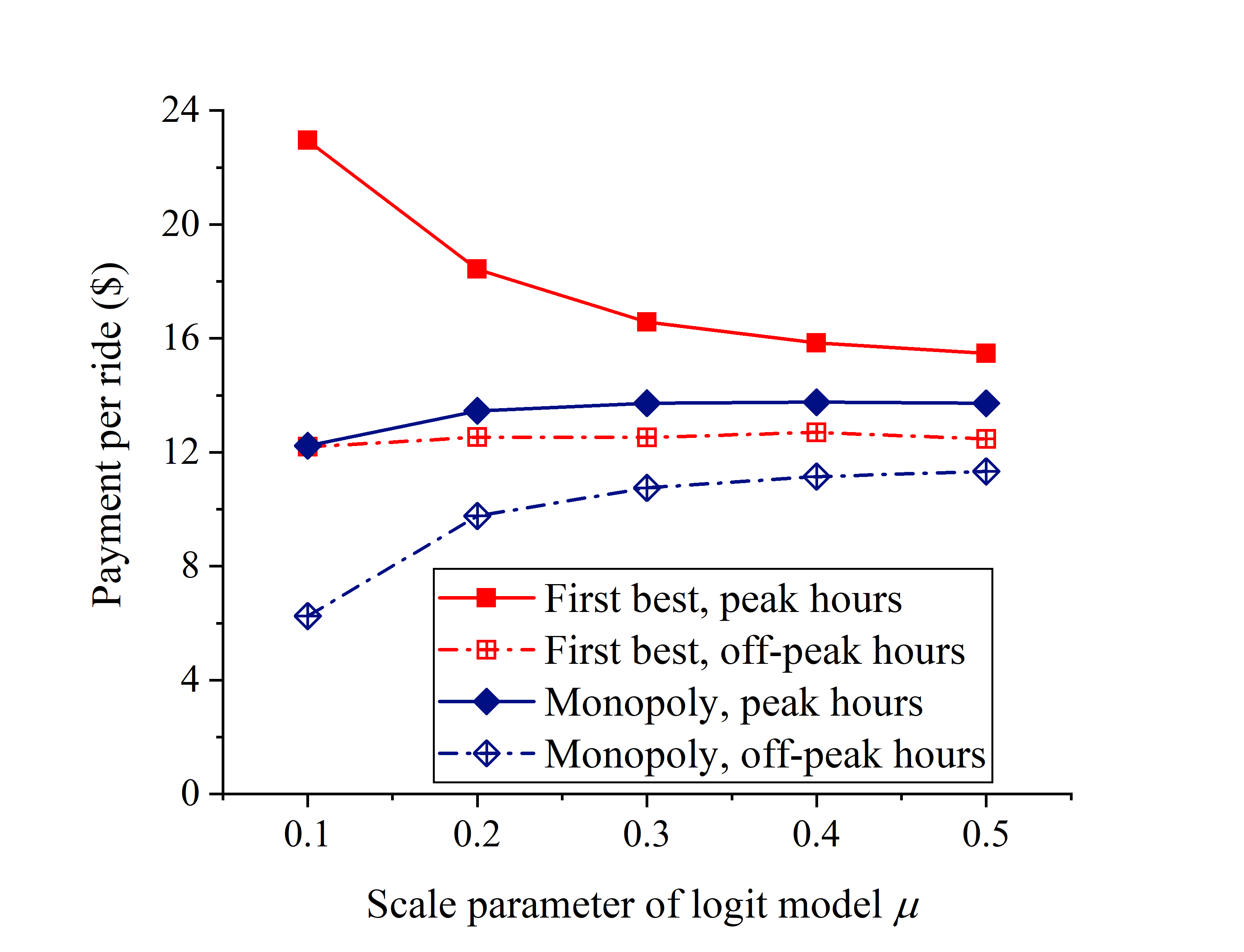}}\\
	\subfloat[][]{\includegraphics[width=0.5\textwidth]{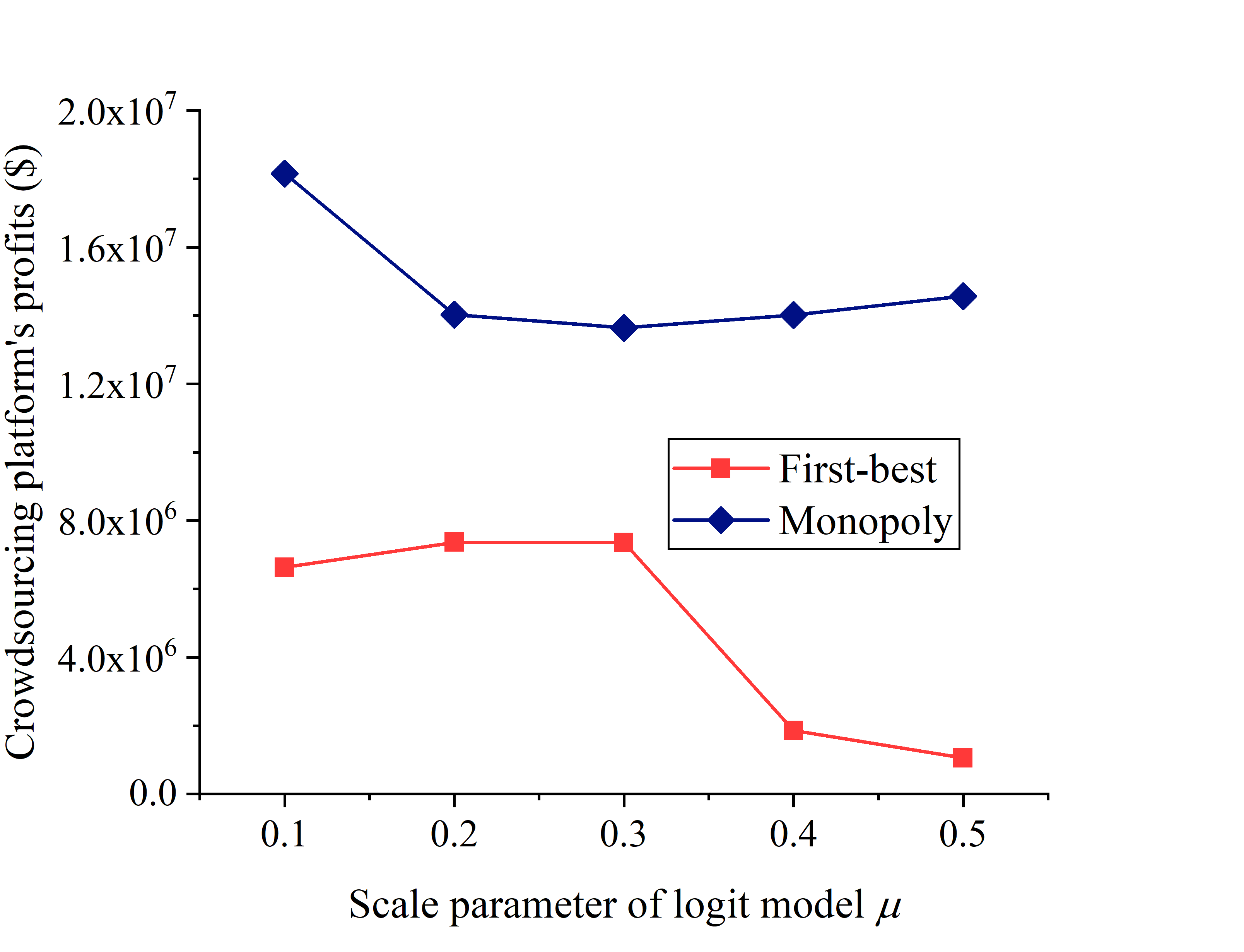}}
	\subfloat[][]{\includegraphics[width=0.5\textwidth]{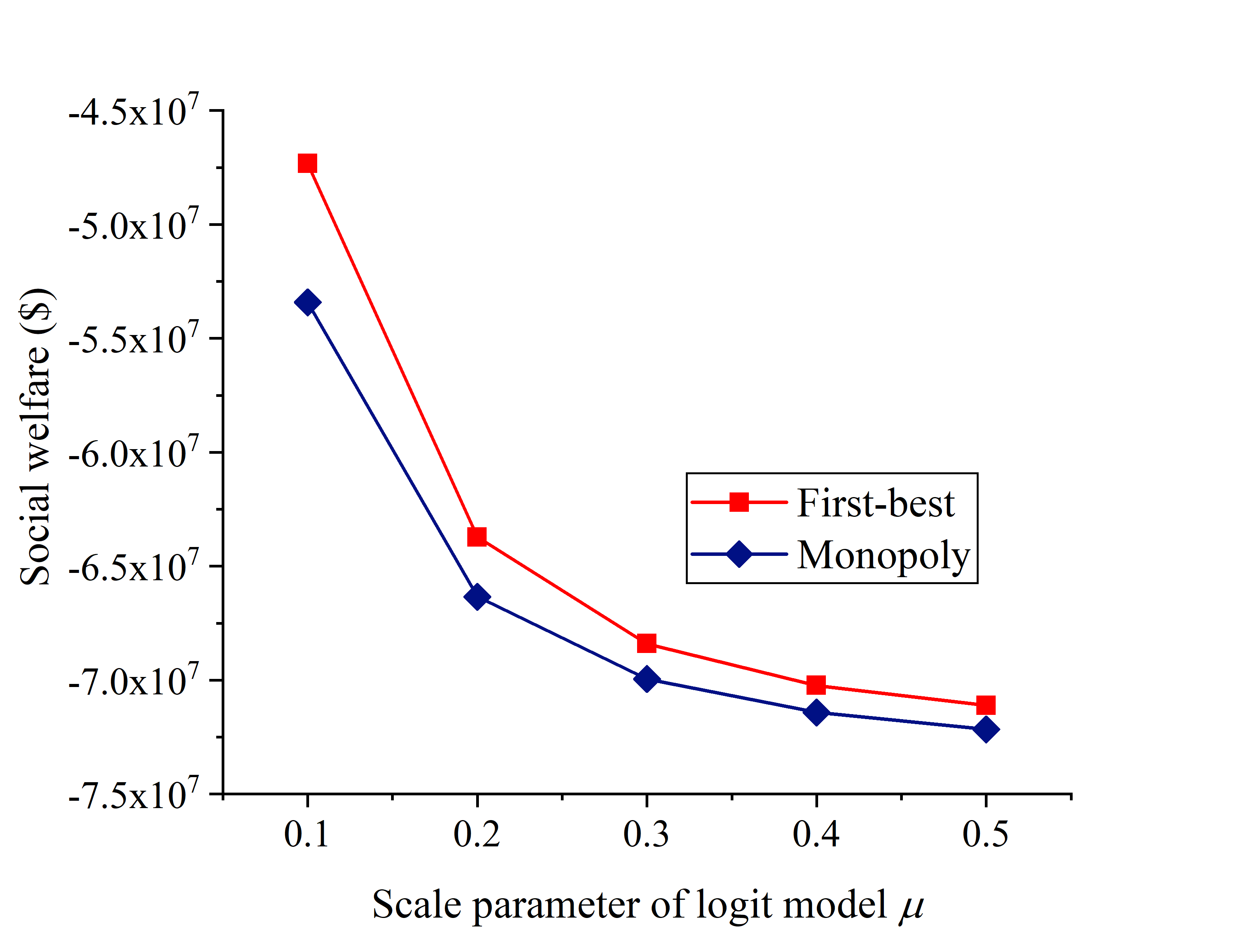}}\\
	\caption[]{Influence of people's sensitivity to utilities.} 
	\label{fig_mu}
\end{figure}
\par
\begin{figure}[h]
	\centering
	\subfloat[][]{\includegraphics[width=0.5\textwidth]{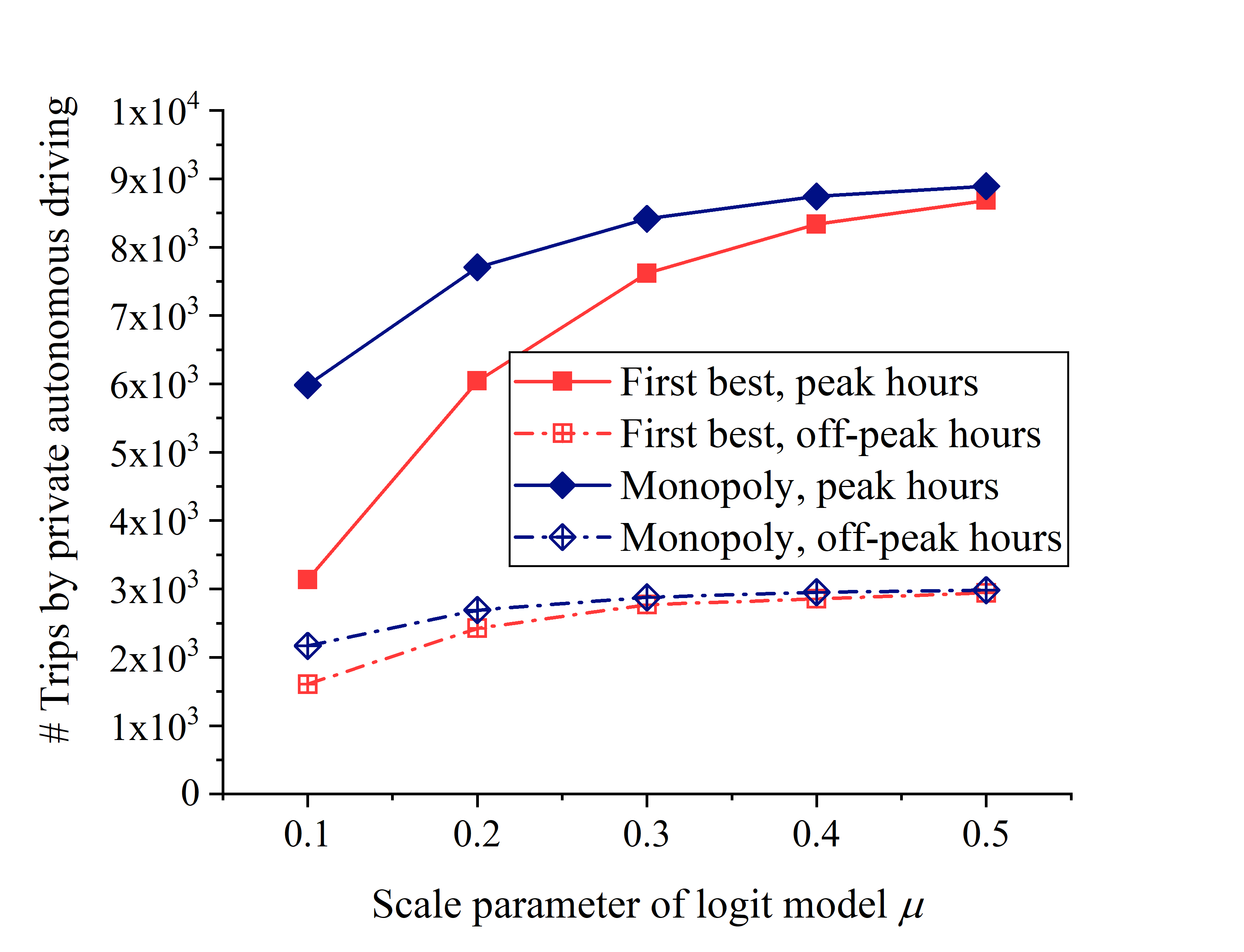}}	
	\subfloat[][]{\includegraphics[width=0.5\textwidth]{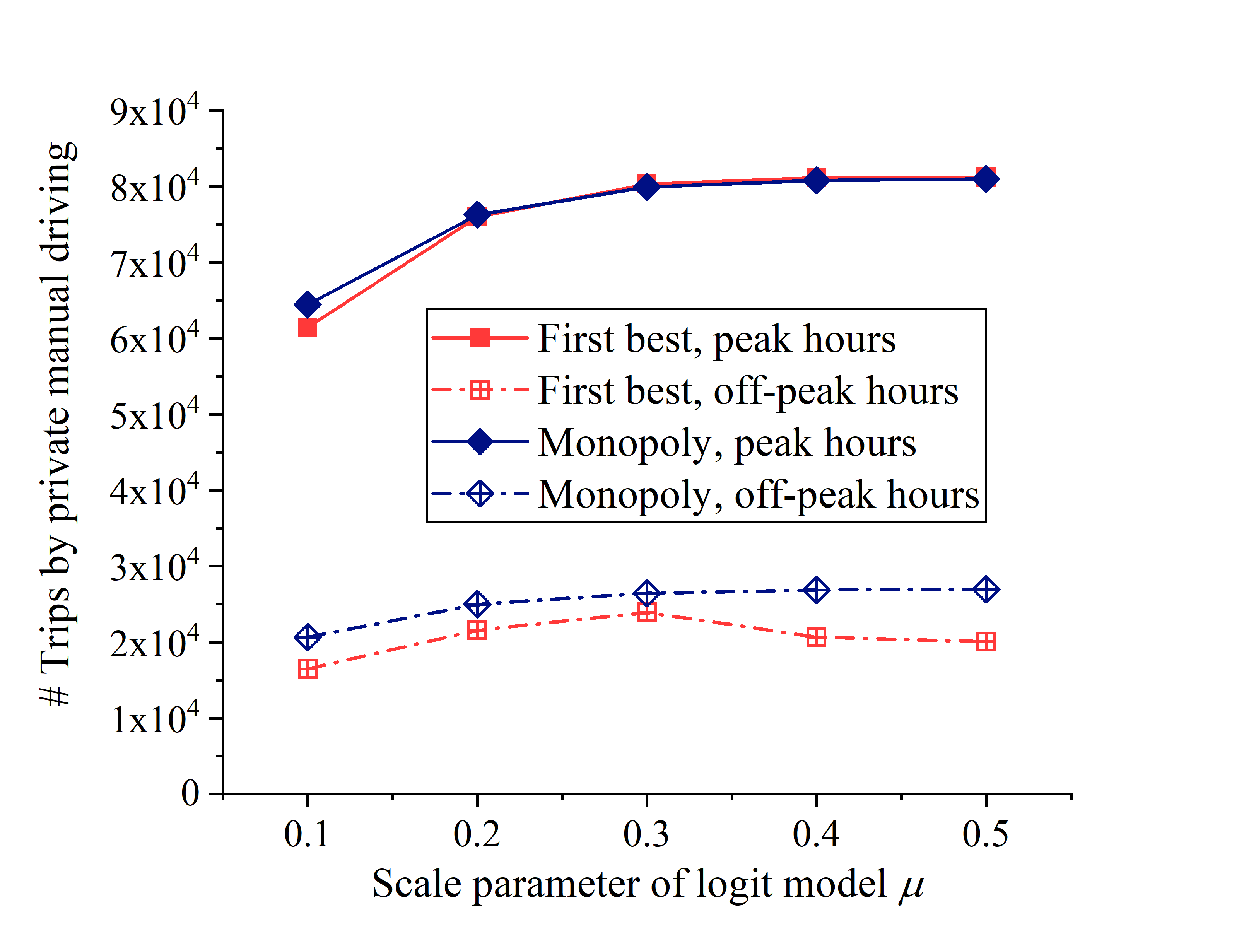}}\\	
	\subfloat[][]{\includegraphics[width=0.5\textwidth]{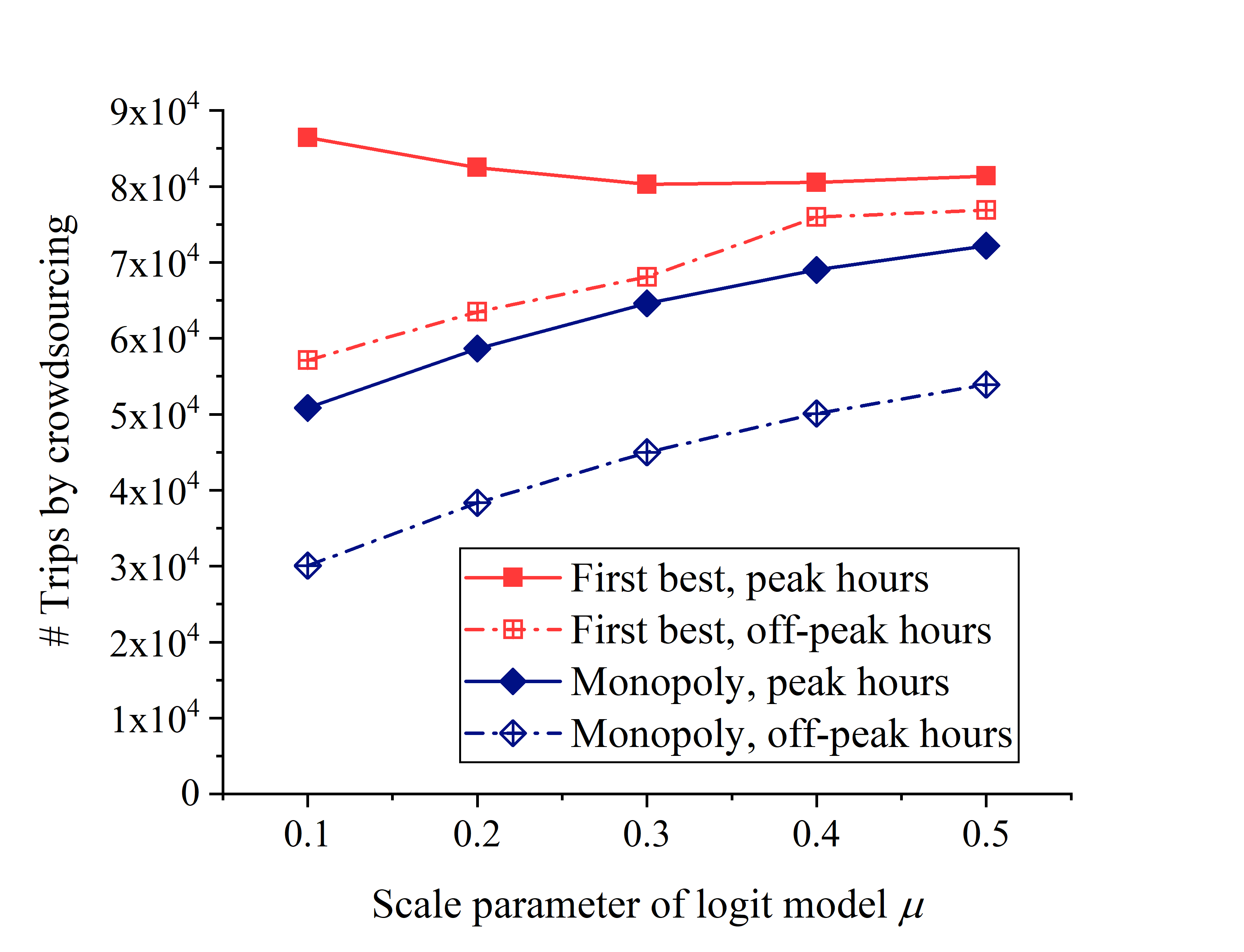}}
	\subfloat[][]{\includegraphics[width=0.5\textwidth]{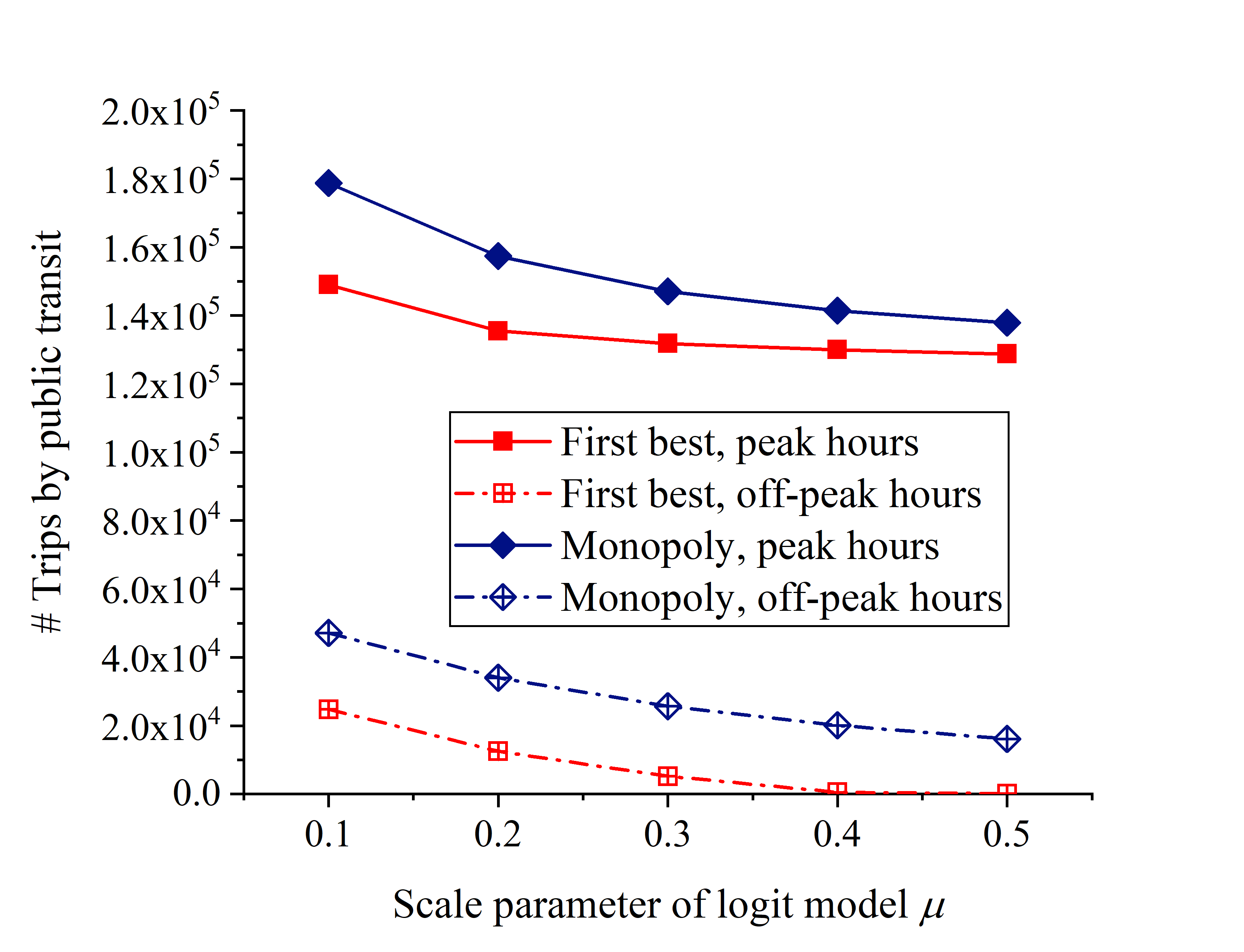}}\\
	\caption[]{Influence of people's sensitivity to utilities on mode split.} 
	\label{fig_mu2}
\end{figure}
\noindent

Without loss of generality, we assume that all classes of citizens are characterized by the same degree of uncertainty with regard to choices, i.e., $\mu_x= \mu \ (\forall x\in \{1,2,...,6\})$. As shown in Figures \ref{fig_mu}(a) and \ref{fig_mu}(b), when people are more sensitive to utilities, it will be harder for the crowdsourcing platform to attract citizens, and thus the trip fare decreases, and the payment increases. It is observed from Figures \ref{fig_mu}(c) and \ref{fig_mu}(d) that the crowdsourcing platform's profit and the social welfare decline overall. Trips by public transit decrease, implying that the shortcoming of its high time cost is amplified with increasing $\mu$, and trips by other three transport modes increase (Figure \ref{fig_mu2}). One noticeable deviation is that the payment during the peak hours for the first-best scenarios decreases (Figure \ref{fig_mu}(b)), leading to slightly fewer crowdsourcing trips and slightly greater platform profit. This indicates that more crowdsourcing trips during the peak hours are discouraged with increasing $\mu$, most likely due to the amplification of the congestion externality. Another deviation is the minor increase of the platform profits for the monopoly scenario for $\mu \ge 0.3$ (Figure \ref{fig_mu}(c)) due to an increasing number of crowdsourcing trips with a small amount of price changes. In short, these experiments show that parameter $\mu$ has a great effect on price setting and the equilibrium, and thus it is important for the crowdsourcing platform to investigate this parameter prior to setting the optimal prices in order to maximize its profits; $\mu$ also influences the urban planner's decision whether to regulate the AV crowdsourcing market, because the first-best scenarios present a more limited improvement with a higher $\mu$, i.e., $11.4\%$ for $\mu=0.1$ and only $1.5\%$ for $\mu=0.5$; thus, regulation will have a limited effect when people are highly sensitive to utilities. We assume $\mu=0.5$ in the following experiments, i.e., approximately $92.4\%$ of people make a choice with $\$5$ higher utility.

\subsubsection{Population density of the city}
\begin{figure}[h]
	\centering
	\subfloat[][]{\includegraphics[width=0.5\textwidth]{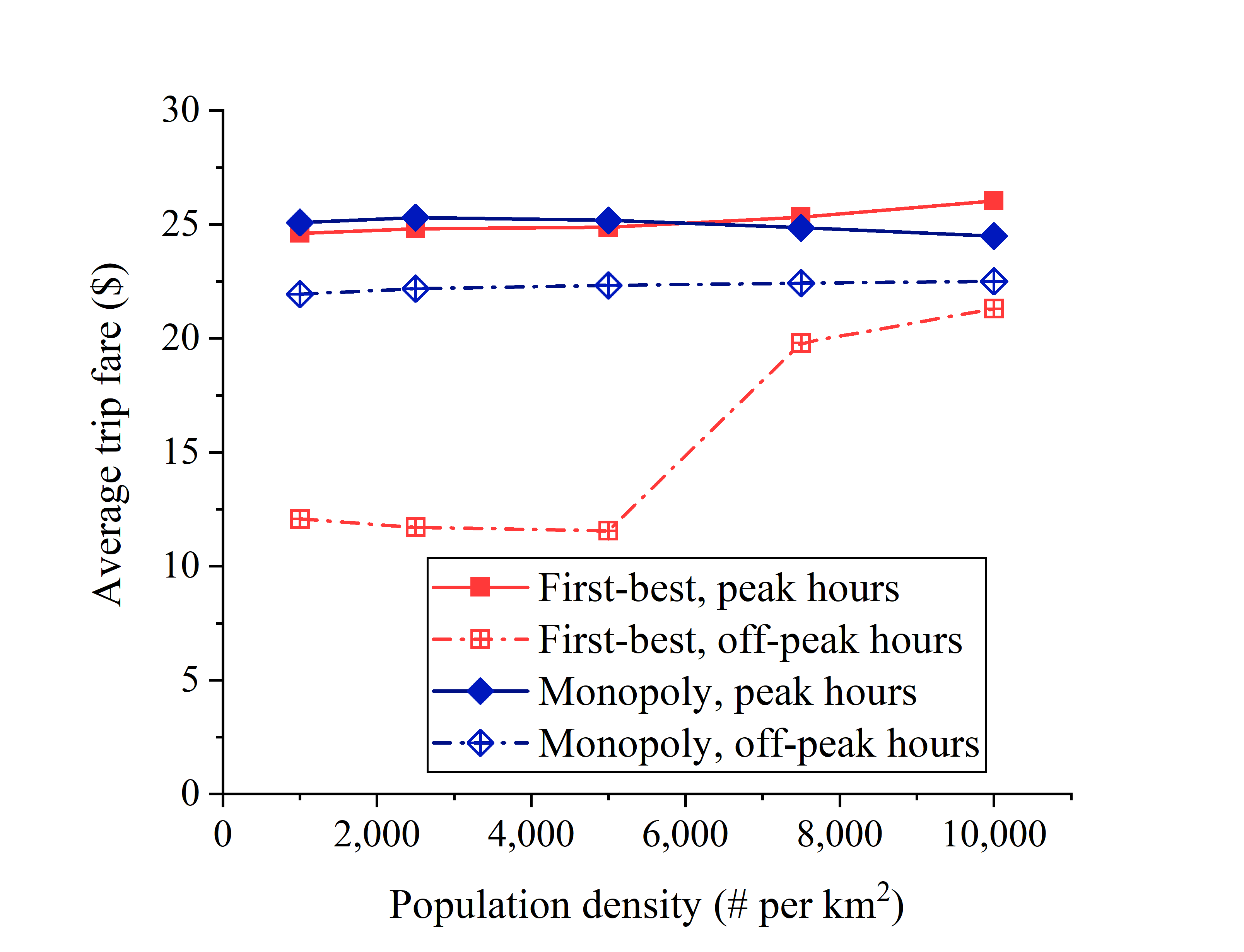}}
	\subfloat[][]{\includegraphics[width=0.5\textwidth]{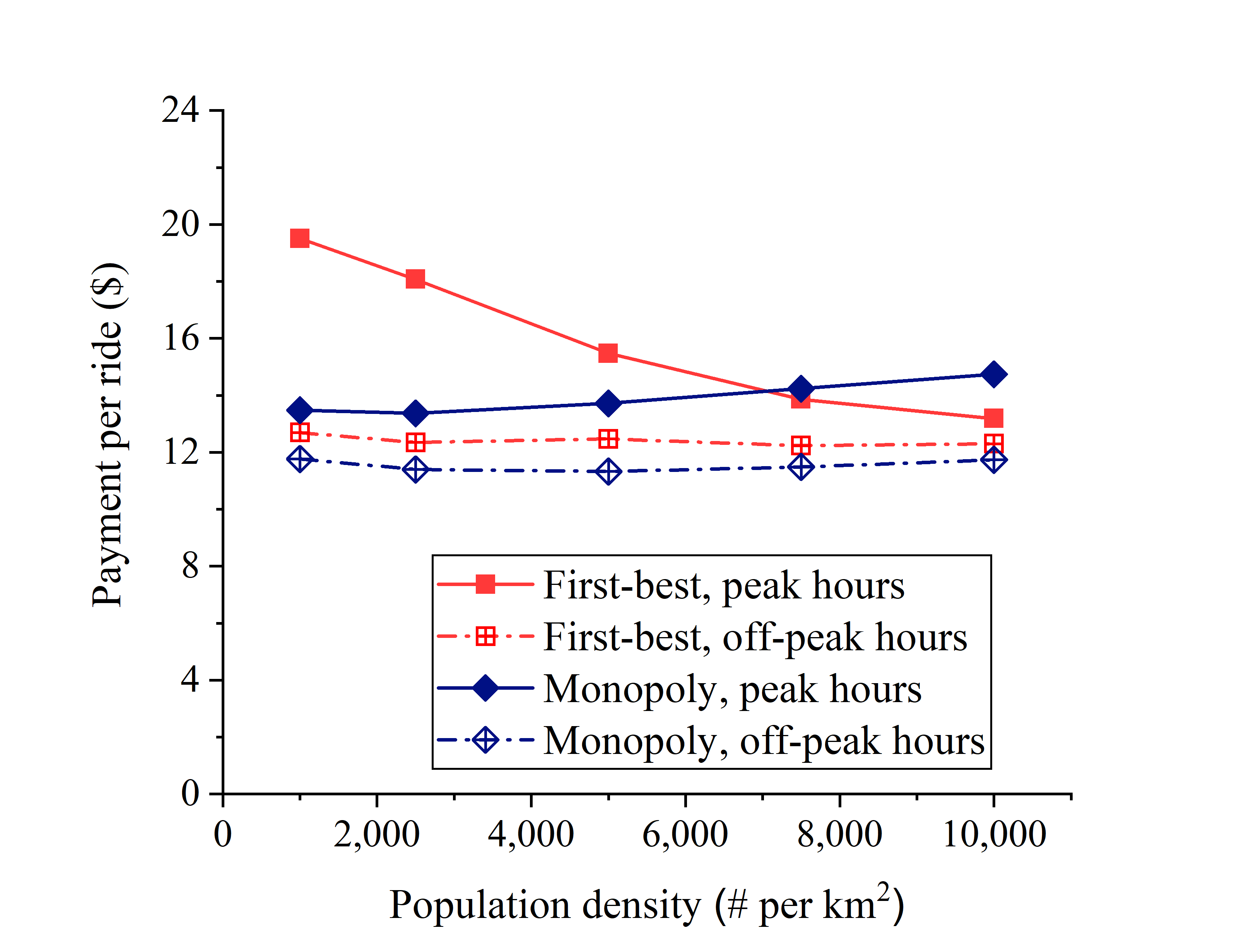}}\\
	\subfloat[][]{\includegraphics[width=0.5\textwidth]{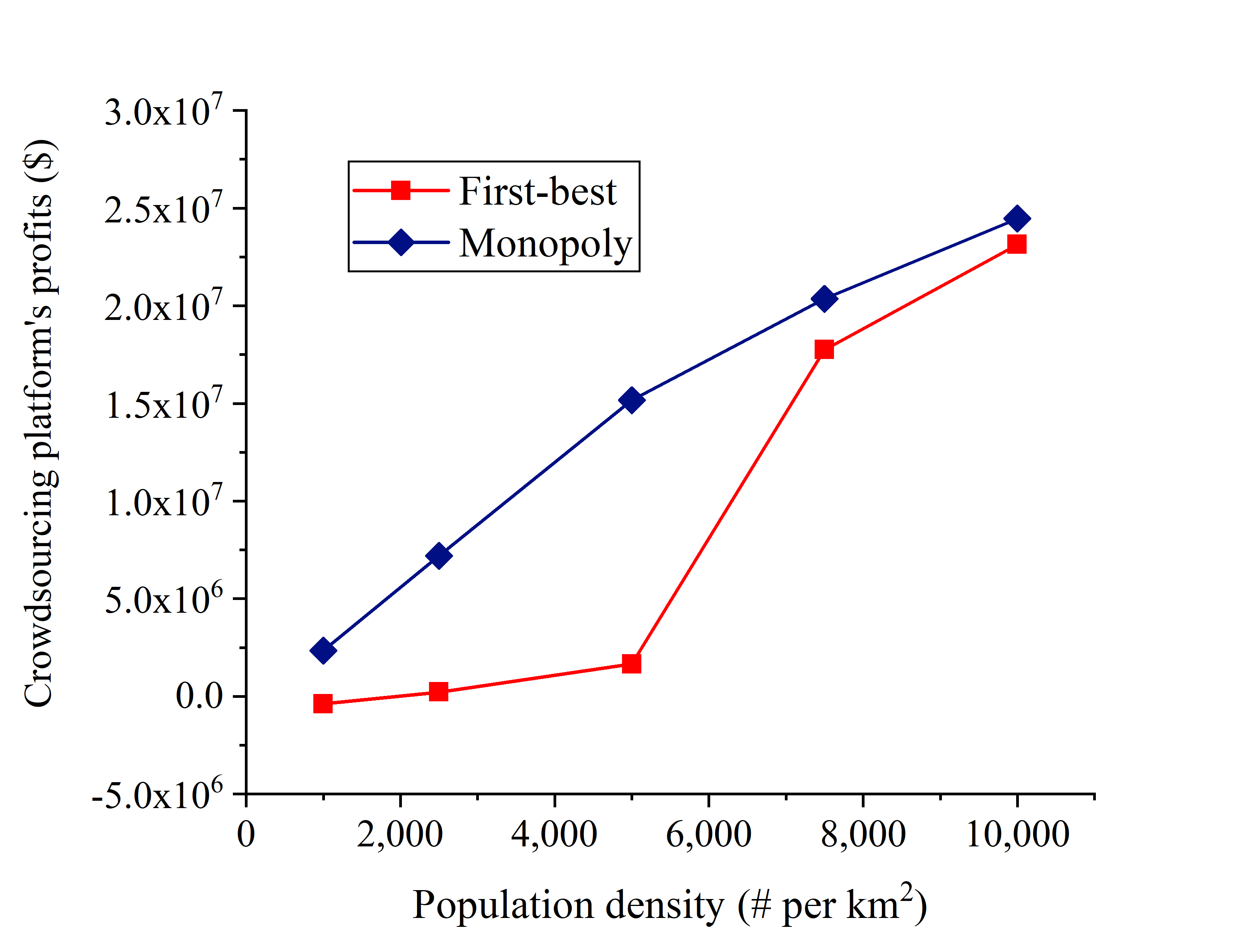}}
	\subfloat[][]{\includegraphics[width=0.5\textwidth]{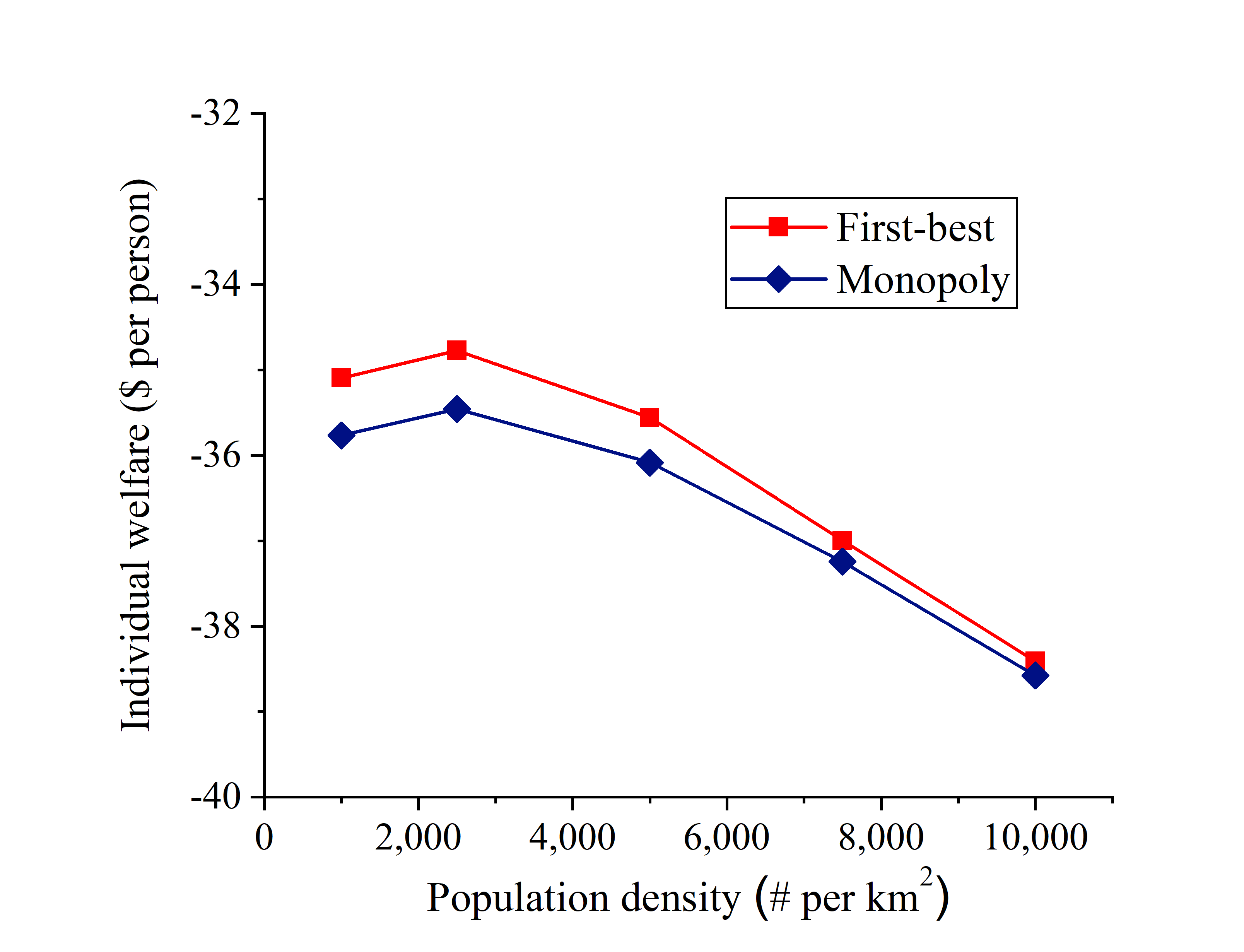}}\\
	\caption[]{Influence of population density.} 
	\label{fig_density}
\end{figure}
\noindent

Figures \ref{fig_density}(a) and \ref{fig_density}(b) show that the population density generally has little impact on prices with exceptions for a few cases, namely the payment during peak-hours declines for the first-best scenarios, and the fare during the off-peak hours increases quickly when the population density grows beyond 5000 persons per $\text{km}^2$. This discourages people to rent AVs or take crowdsourcing service and thus alleviate the congestion externality that grows with the population density. This implies that main impact of the population density on the equilibrium is the congestion externality. Figure \ref{fig_density}(c) shows that the crowdsourcing platform's profits increase, indicating that the crowdsourcing platform must expect to be launched in more densely populated areas, such as a metropolis. Figure \ref{fig_density}(d) shows that individual welfare first increases slightly and then decreases, most likely because a low population density induces insufficient number of crowdsourcing AVs, increasing the individual travel cost. The highest individual welfare is obtained when the population density is approximately $2500$ persons per $\text{km}^2$.

\subsubsection{AV market penetration rate}
We denote the autonomous vehicle market penetration rate as the ratio of the population who own AVs to the population who own AVs or MVs. A greater market penetration rate means a larger pool for AV crowdsourcing. As shown in Figures \ref{fig_penetration}(a) and \ref{fig_penetration}(b), fare and payment decline due to a more sufficient supply and the rate of decline decreases with growing AV market penetration rate. As observed from Figure \ref{fig_penetration}(c), the crowdsourcing platform's profits increase with AV popularization, demonstrating the economy of scale. Figure \ref{fig_penetration}(d) demonstrates that social welfare increase significantly, suggesting that the urban planner should have a supportive attitude toward AV popularization. Moreover, the wider gap in the social welfare between the first-best and monopoly scenarios suggests that regulation will become increasingly beneficial with increasing AV market penetration.  
 
 \begin{figure}[h]
	\centering
	\subfloat[][]{\includegraphics[width=0.5\textwidth]{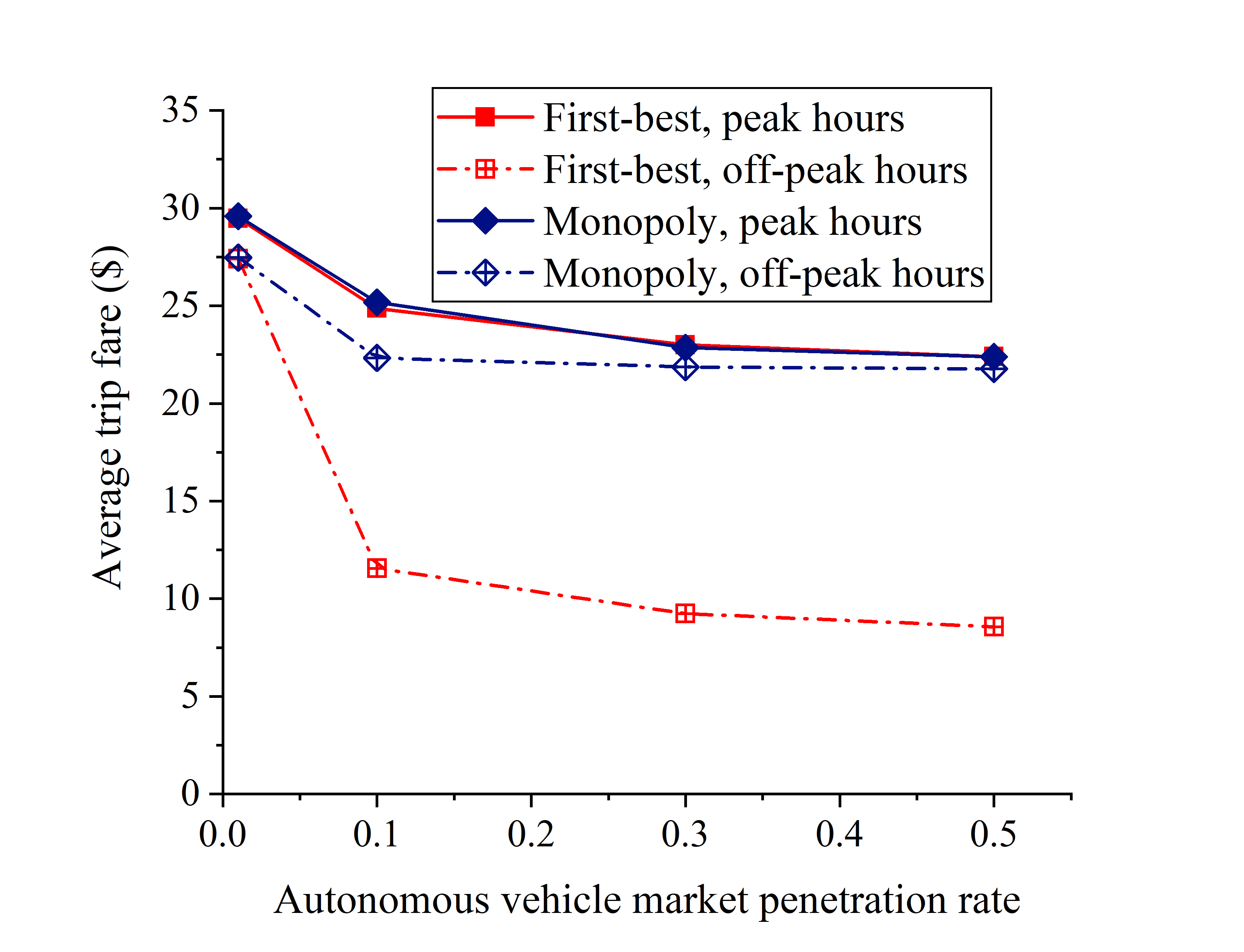}}	
	\subfloat[][]{\includegraphics[width=0.5\textwidth]{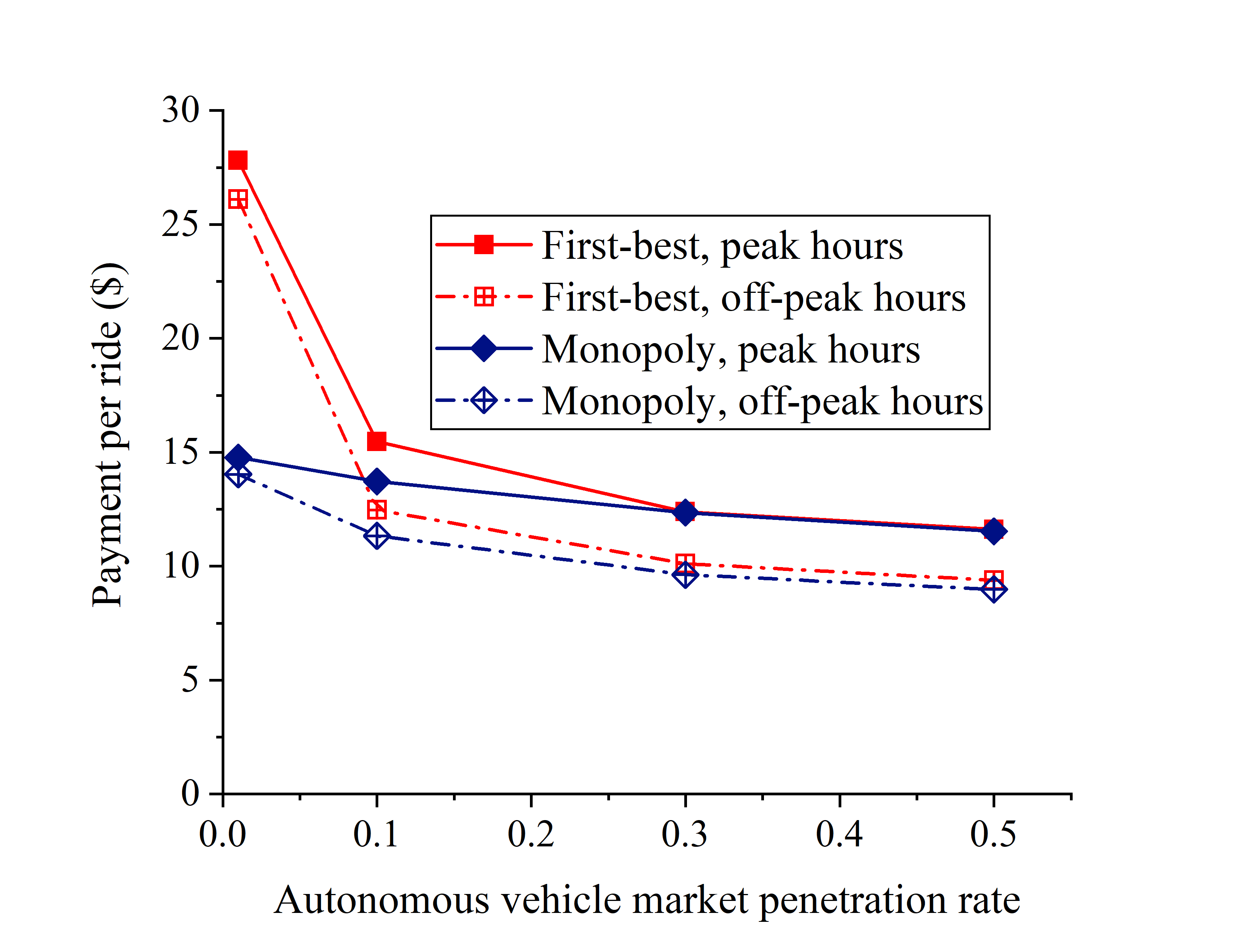}}\\	
	\subfloat[][]{\includegraphics[width=0.5\textwidth]{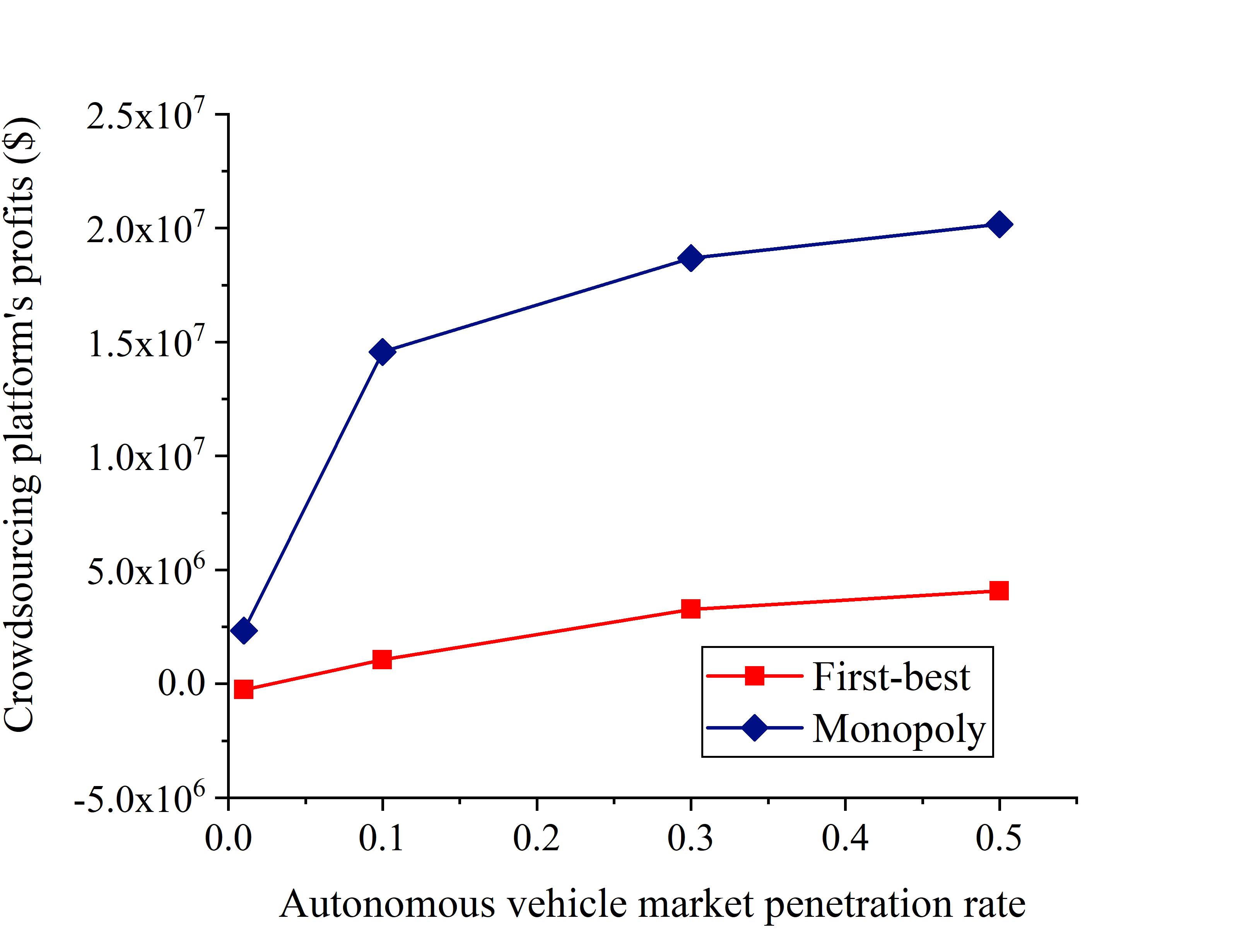}}
	\subfloat[][]{\includegraphics[width=0.5\textwidth]{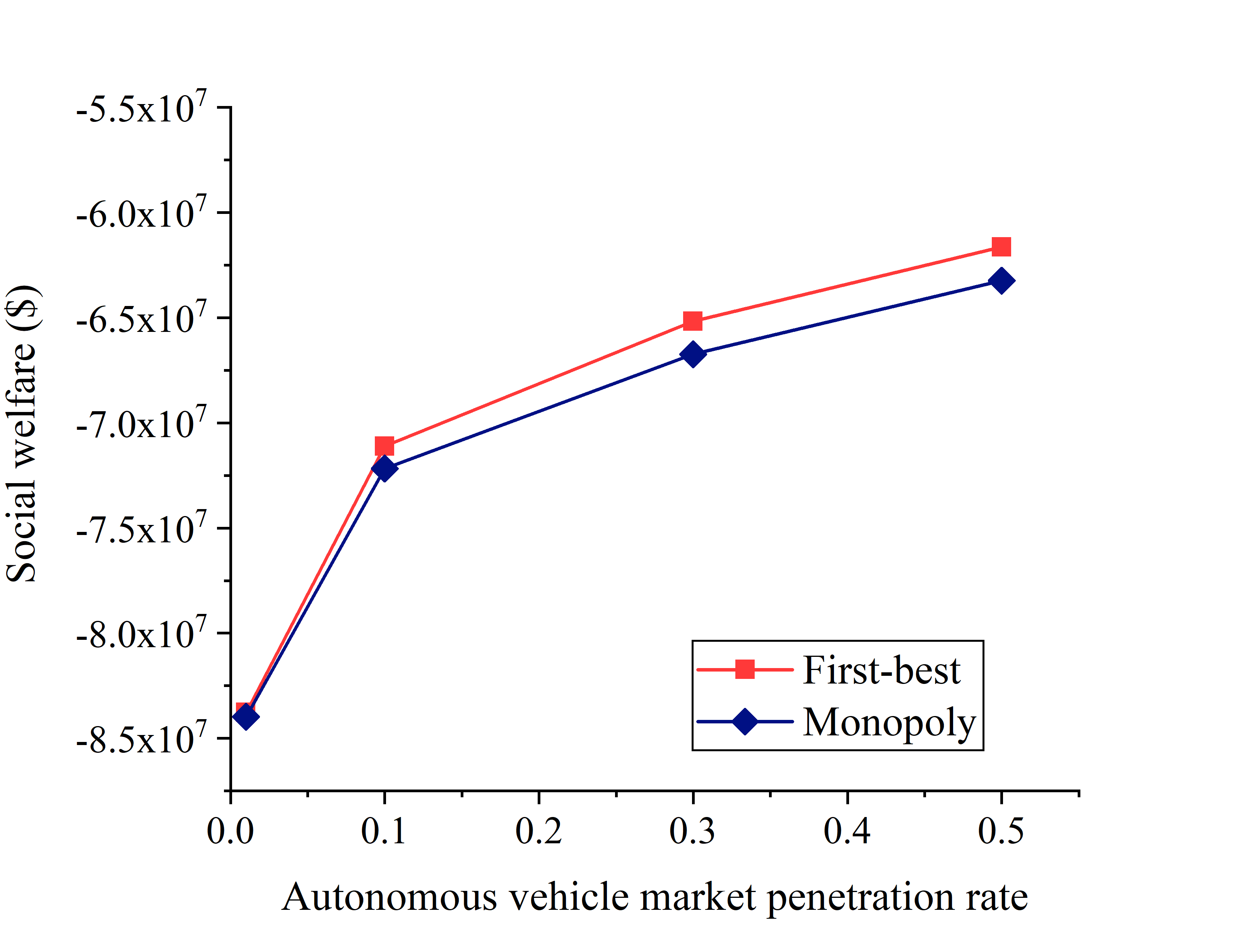}}\\
	\caption[]{Influence of the autonomous vehicle market penetration rate.} 
	\label{fig_penetration}
\end{figure}
\par

\subsubsection{AV technology maturity}
The maturity of AV technology $(1-\alpha)$ represents the degree of improvement on road capacity, where $\alpha$ represents the "relative occupation" of AVs compared to MVs. As Figures \ref{fig_maturity}(a) and \ref{fig_maturity}(b) show, AV technology maturity's impact on prices is ignorable except that the platform raises payment in peak hours for the first-best scenarios to encourage more crowdsourcing trips, which results from higher road capacity. Figure \ref{fig_maturity}(c) shows the crowdsourcing platform's profits are improved in the monopoly scenario, but are deteriorated under regulation. Figure \ref{fig_maturity}(d) points out that the advances of AV technology help improve the social welfare. However, the maturity of AV technology has limited impact on price setting and the equilibrium state.
\begin{figure}[h]
	\centering
	\subfloat[][]{\includegraphics[width=0.5\textwidth]{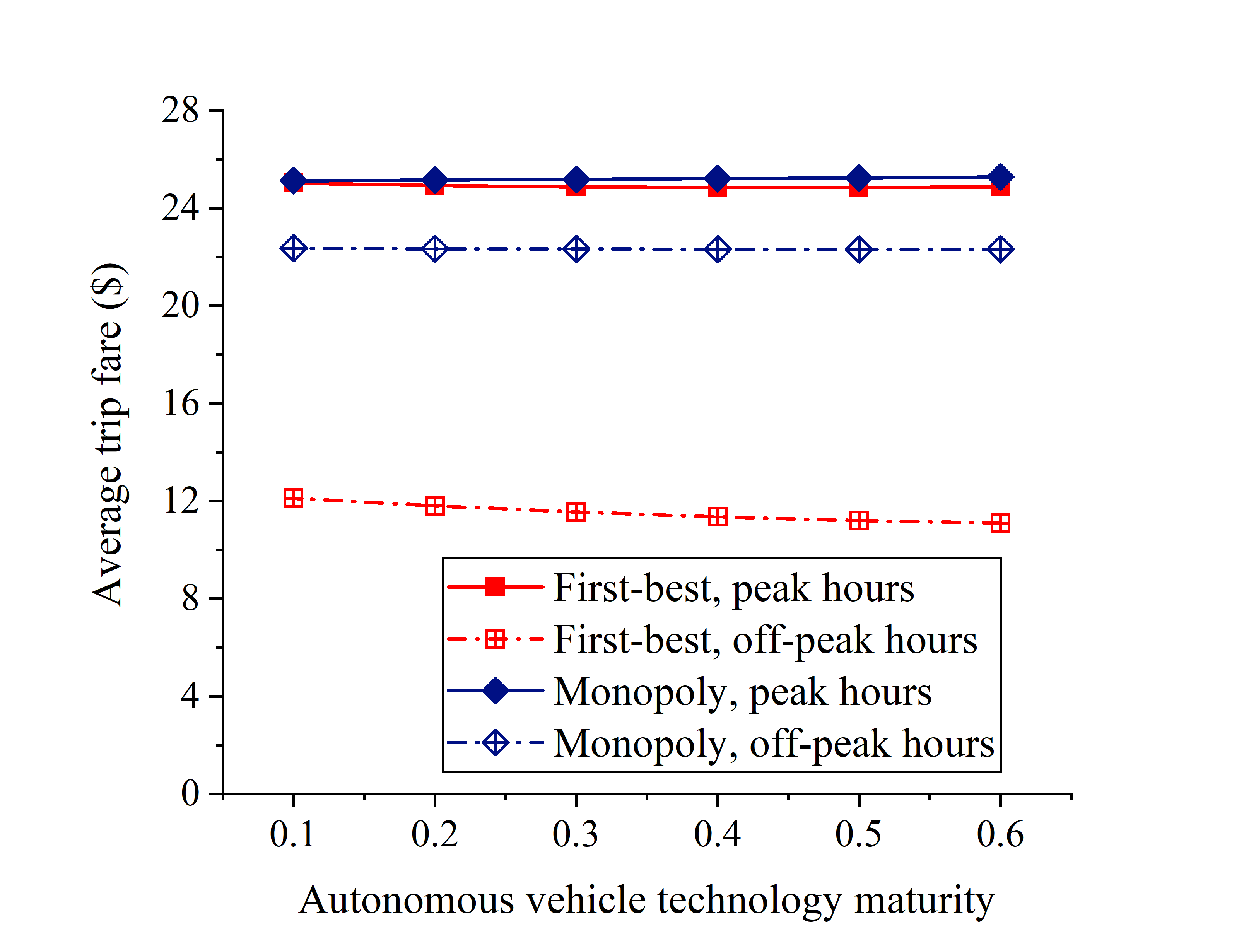}}	
	\subfloat[][]{\includegraphics[width=0.5\textwidth]{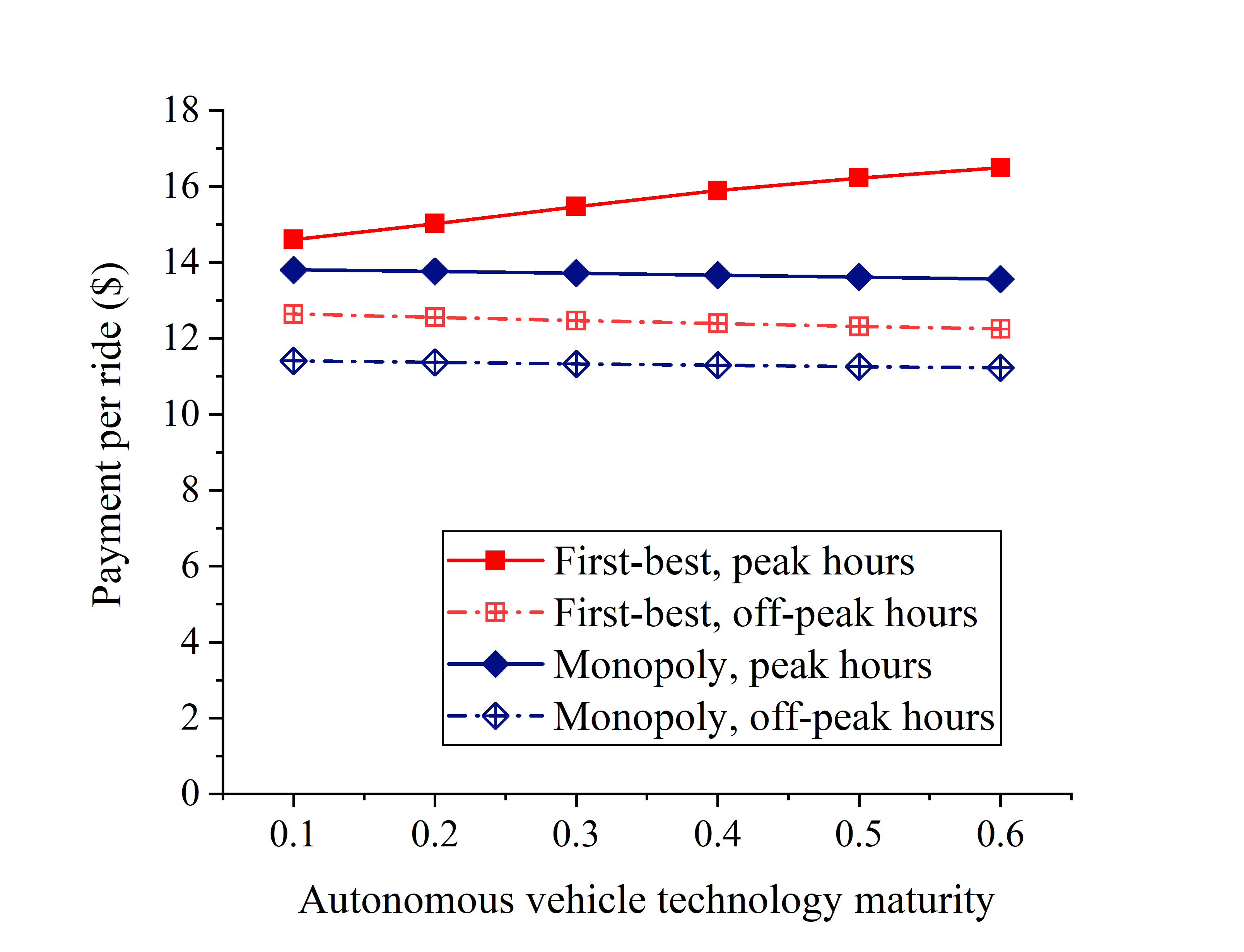}}\\	
	\subfloat[][]{\includegraphics[width=0.5\textwidth]{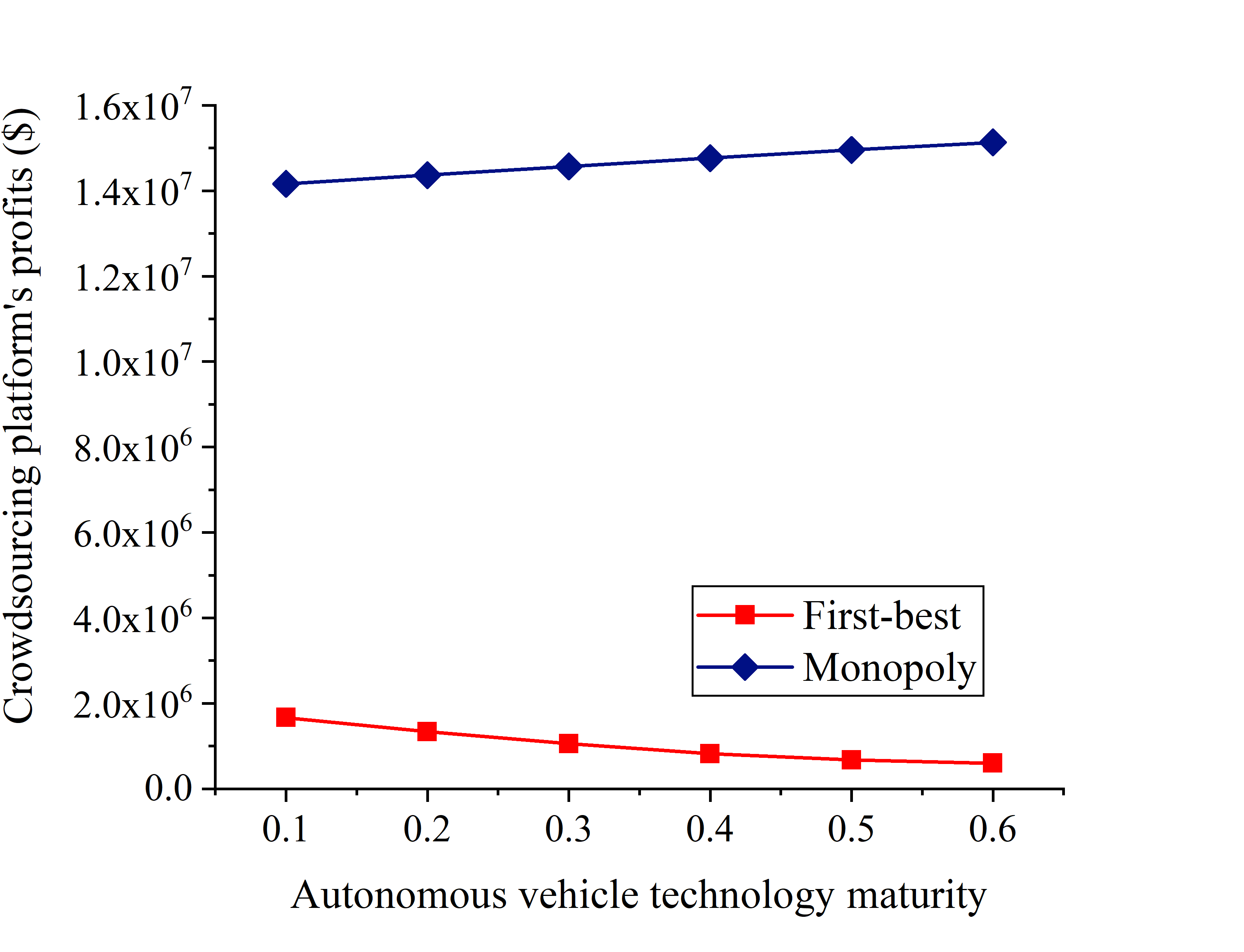}}
	\subfloat[][]{\includegraphics[width=0.5\textwidth]{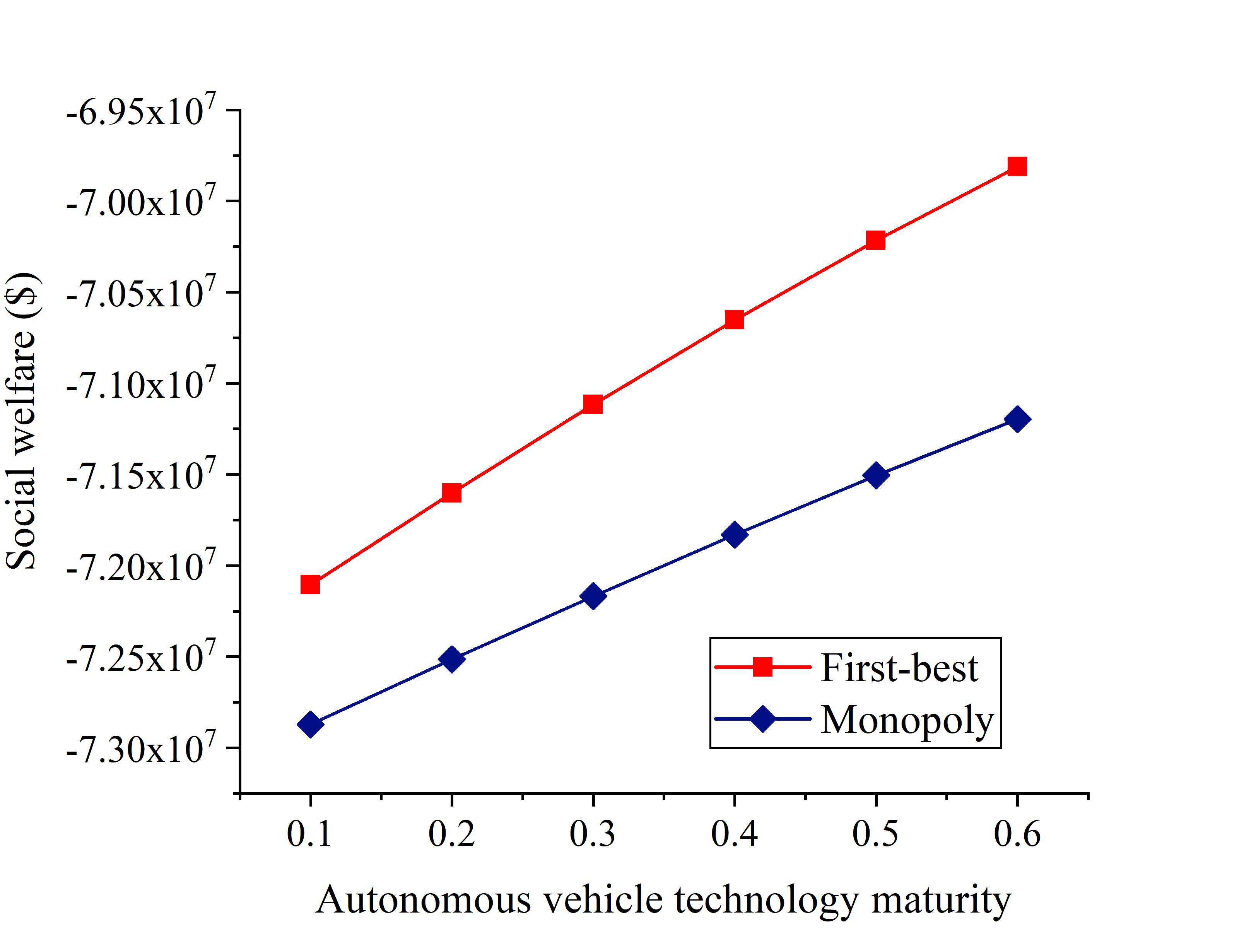}}\\
	\caption[]{Influence of autonomous vehicle technology maturity.} 
	\label{fig_maturity}
\end{figure}
\par
\subsubsection{Number of AVs pre-purchased by the crowdsourcing platform}

As shown by Figures \ref{fig_Ns}(a) and \ref{fig_Ns}(b), the optimal fare and payment remain almost unchanged. The results presented in Figure \ref{fig_Ns}(c) indicate that the crowdsourcing platform should purchase as many AVs as possible to obtain a higher daily profit. An examination of Figure \ref{fig_Ns}(d) shows that social welfare improves with greater number of purchased AVs. Nevertheless, the impact of the number of vehicles pre-purchased is much smaller than the impact of the AV market penetration rate on the equilibrium state.

\begin{figure}[h]
	\centering
	\subfloat[][]{\includegraphics[width=0.5\textwidth]{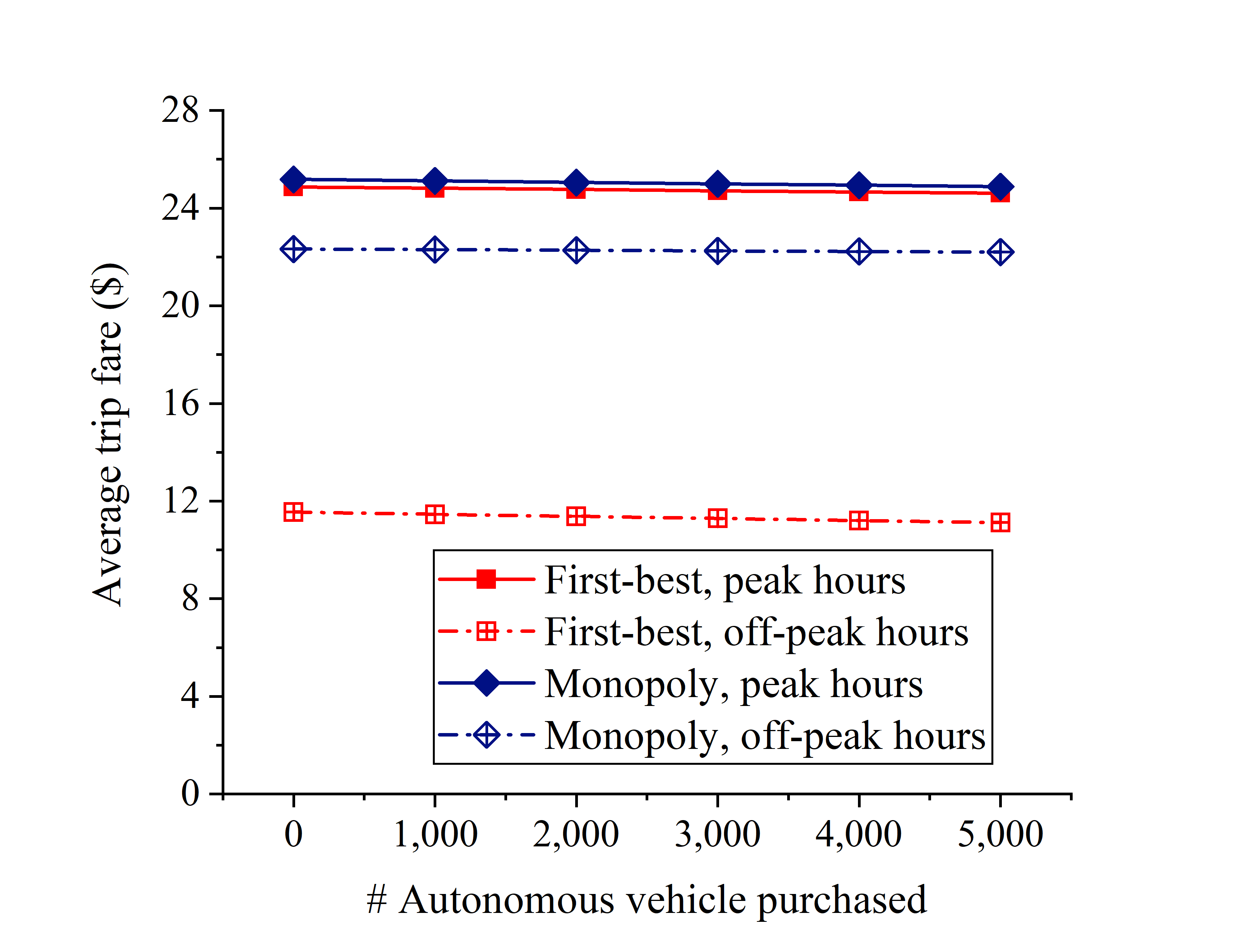}}	
	\subfloat[][]{\includegraphics[width=0.5\textwidth]{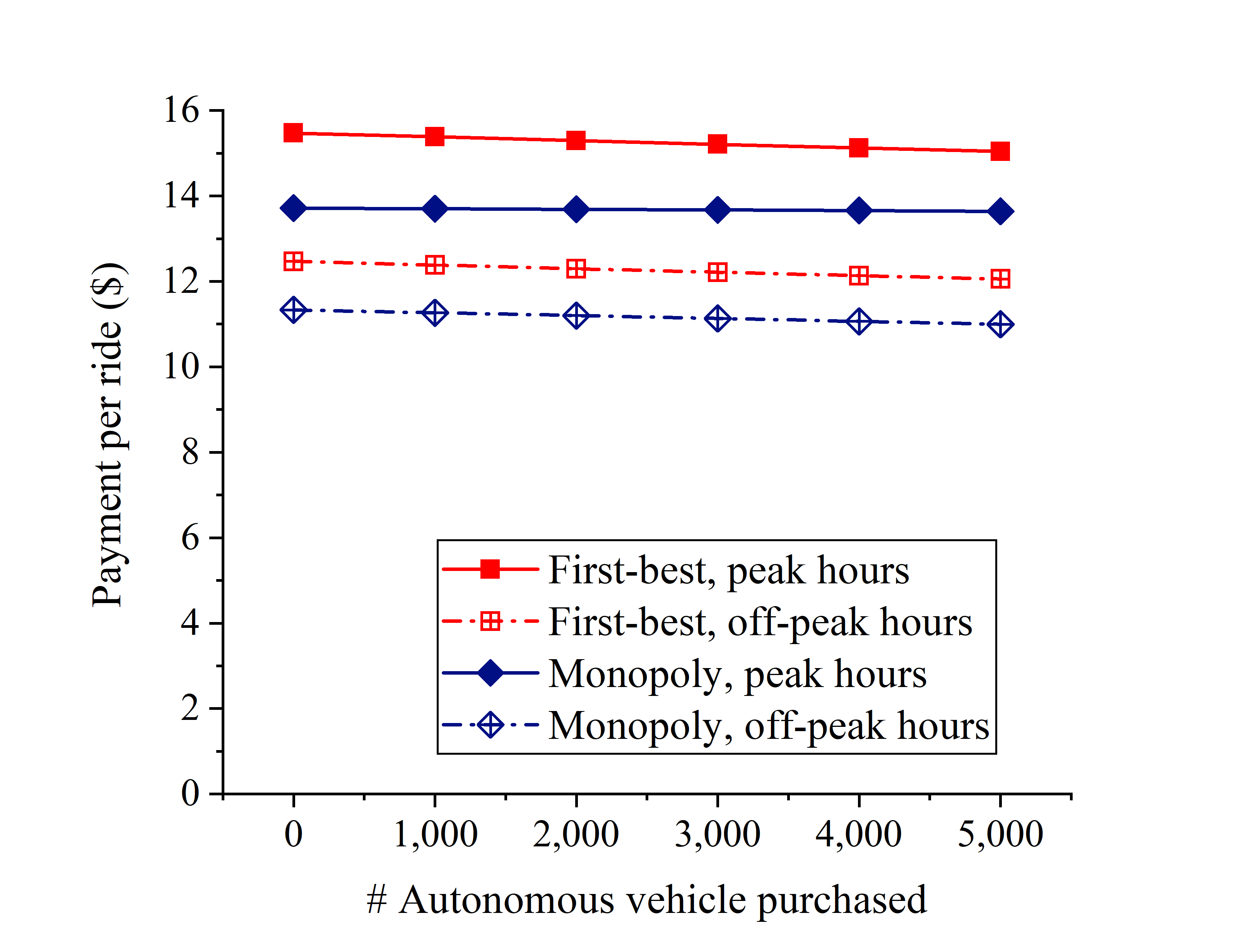}}\\	
	\subfloat[][]{\includegraphics[width=0.5\textwidth]{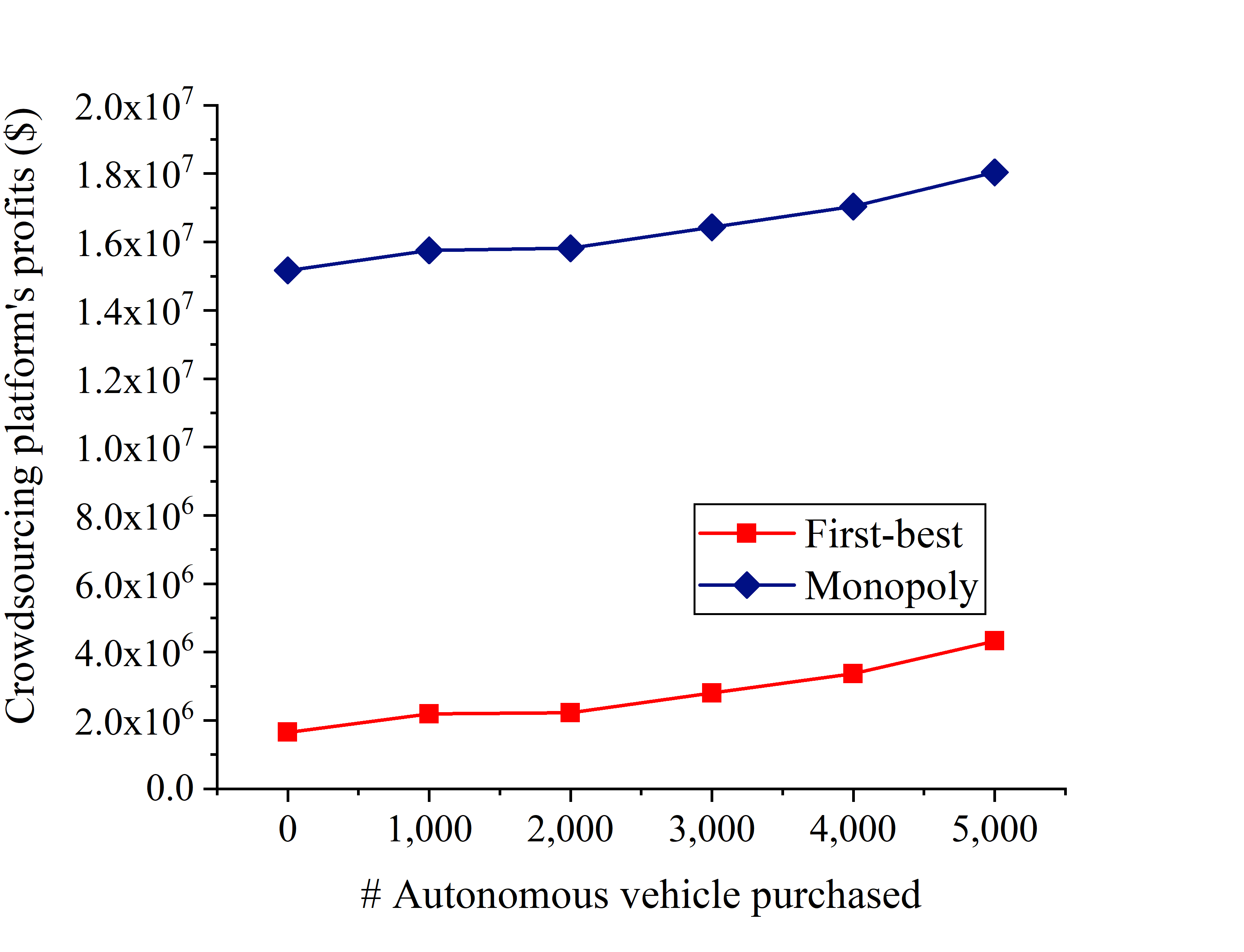}}
	\subfloat[][]{\includegraphics[width=0.5\textwidth]{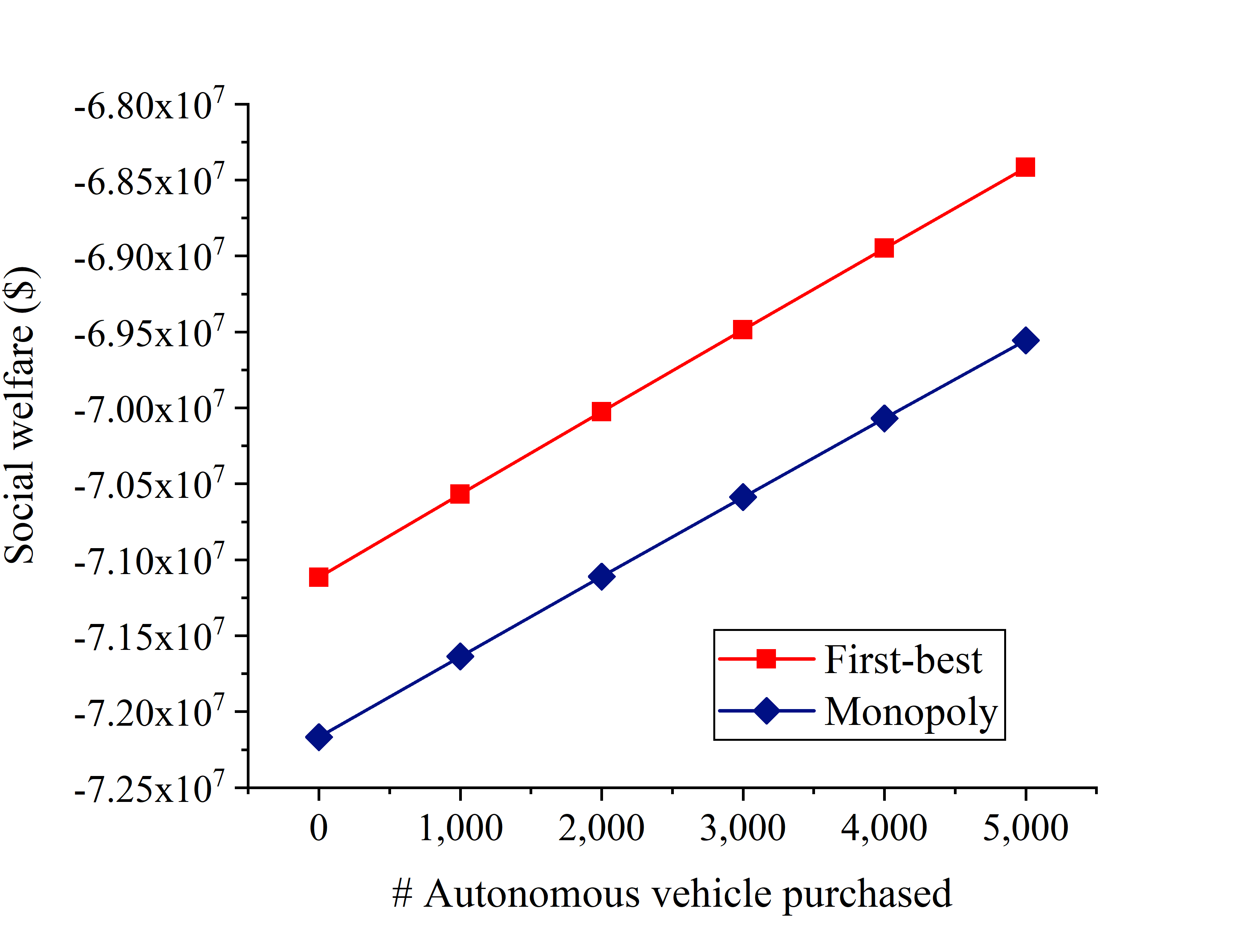}}\\
	\caption[]{Influence of the number of autonomous vehicle purchased.} 
	\label{fig_Ns}
\end{figure}
\par

\begin{figure}[]
	\centering
	\subfloat[][]{\includegraphics[width=0.5\textwidth]{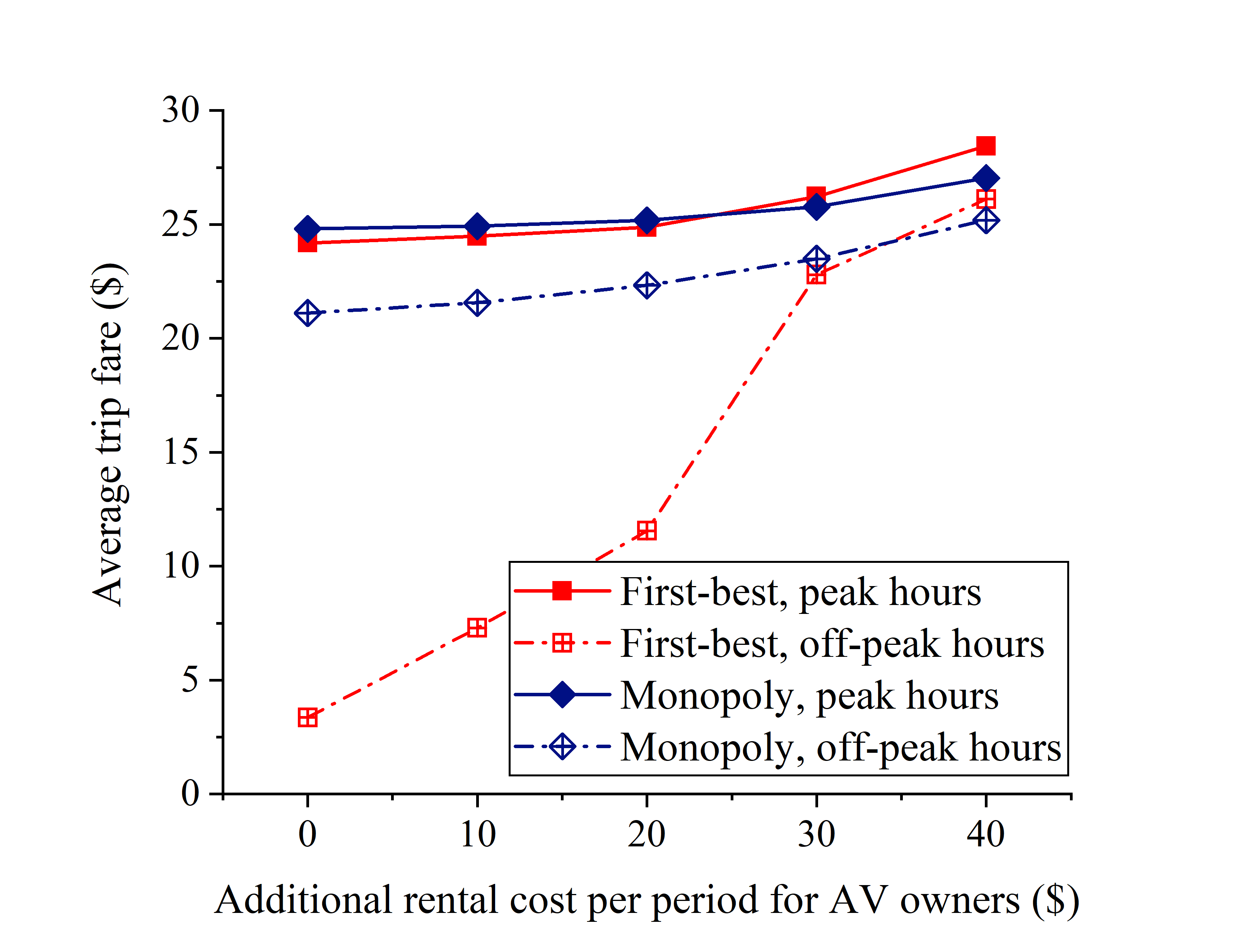}}	
	\subfloat[][]{\includegraphics[width=0.5\textwidth]{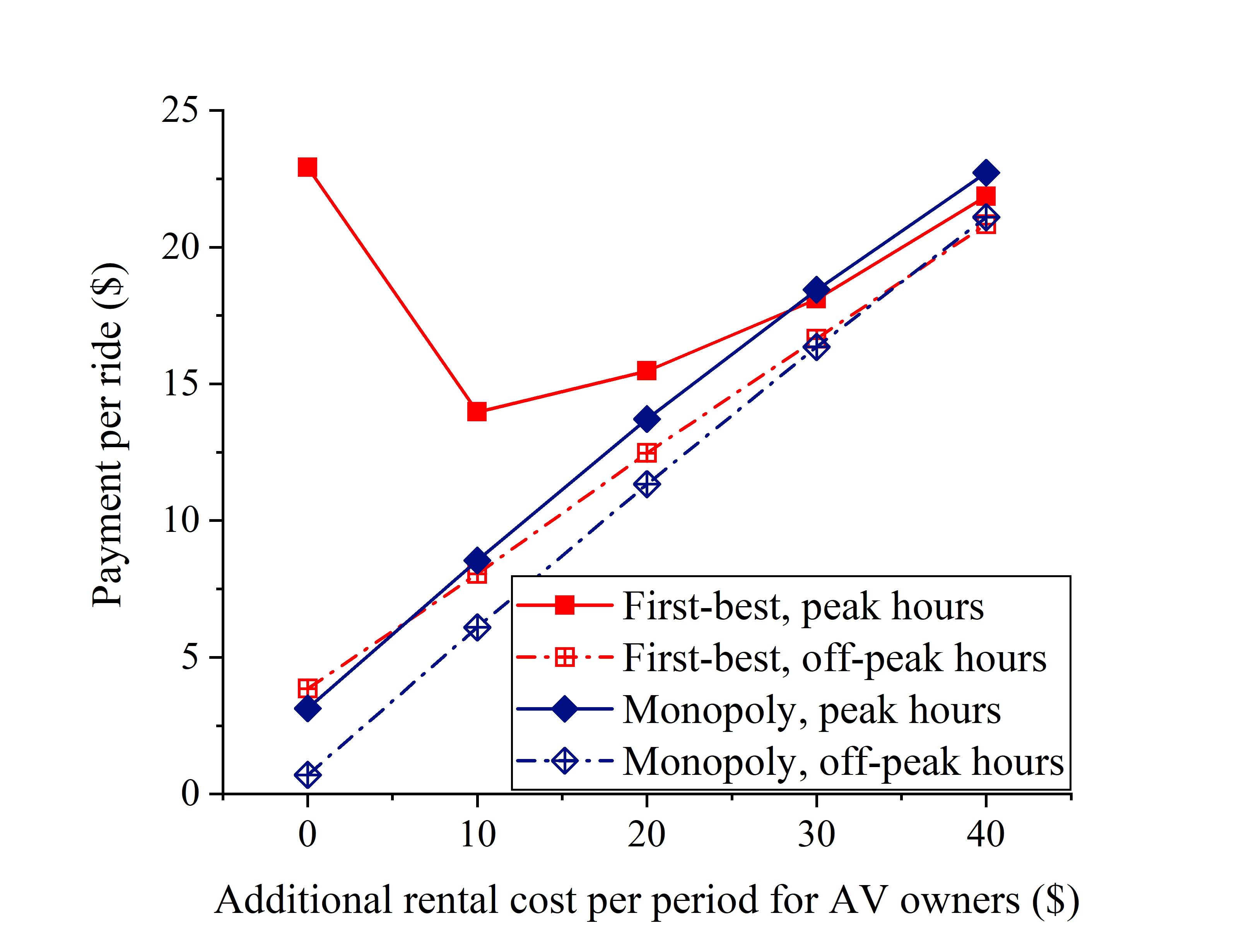}}\\	
	\subfloat[][]{\includegraphics[width=0.5\textwidth]{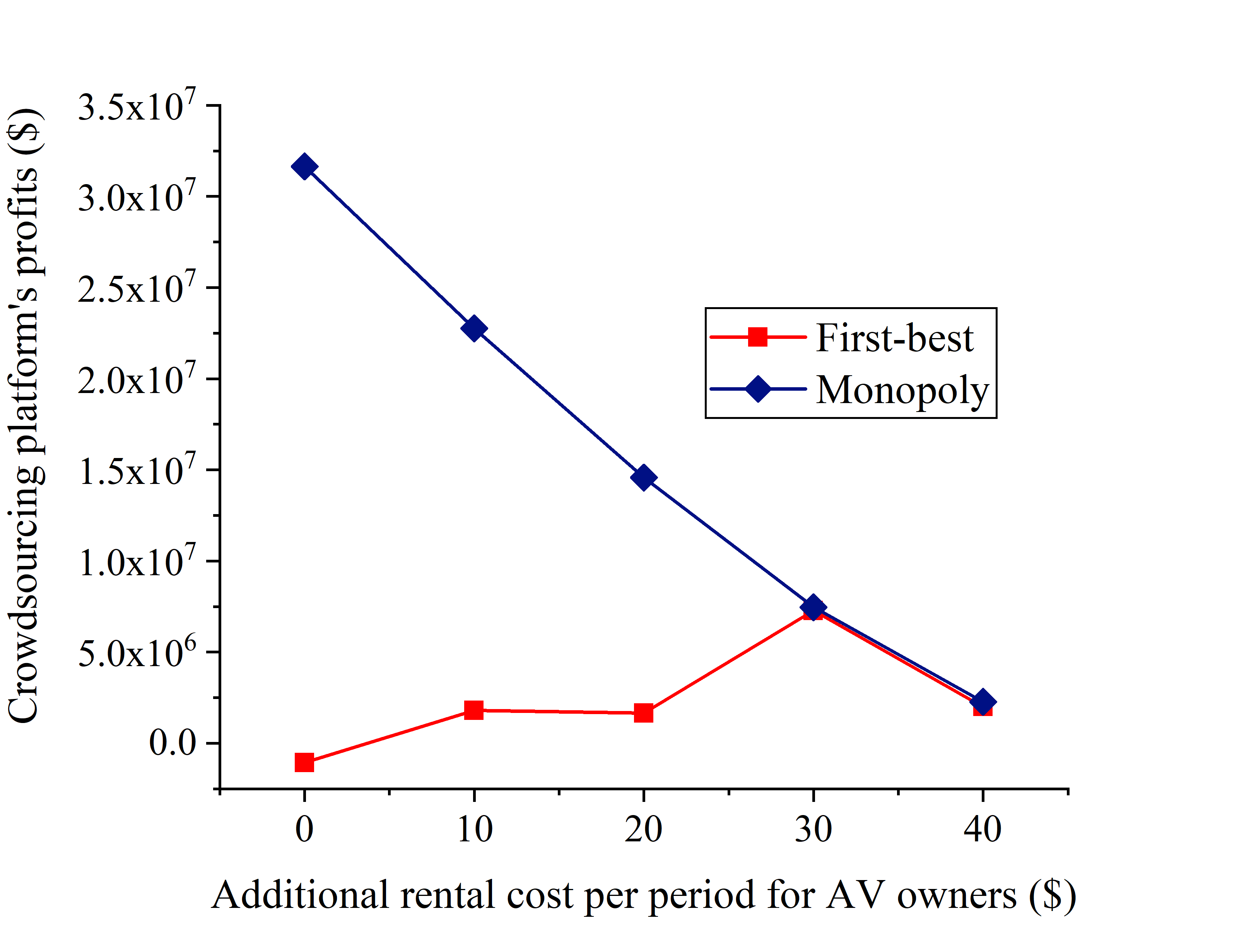}}
	\subfloat[][]{\includegraphics[width=0.5\textwidth]{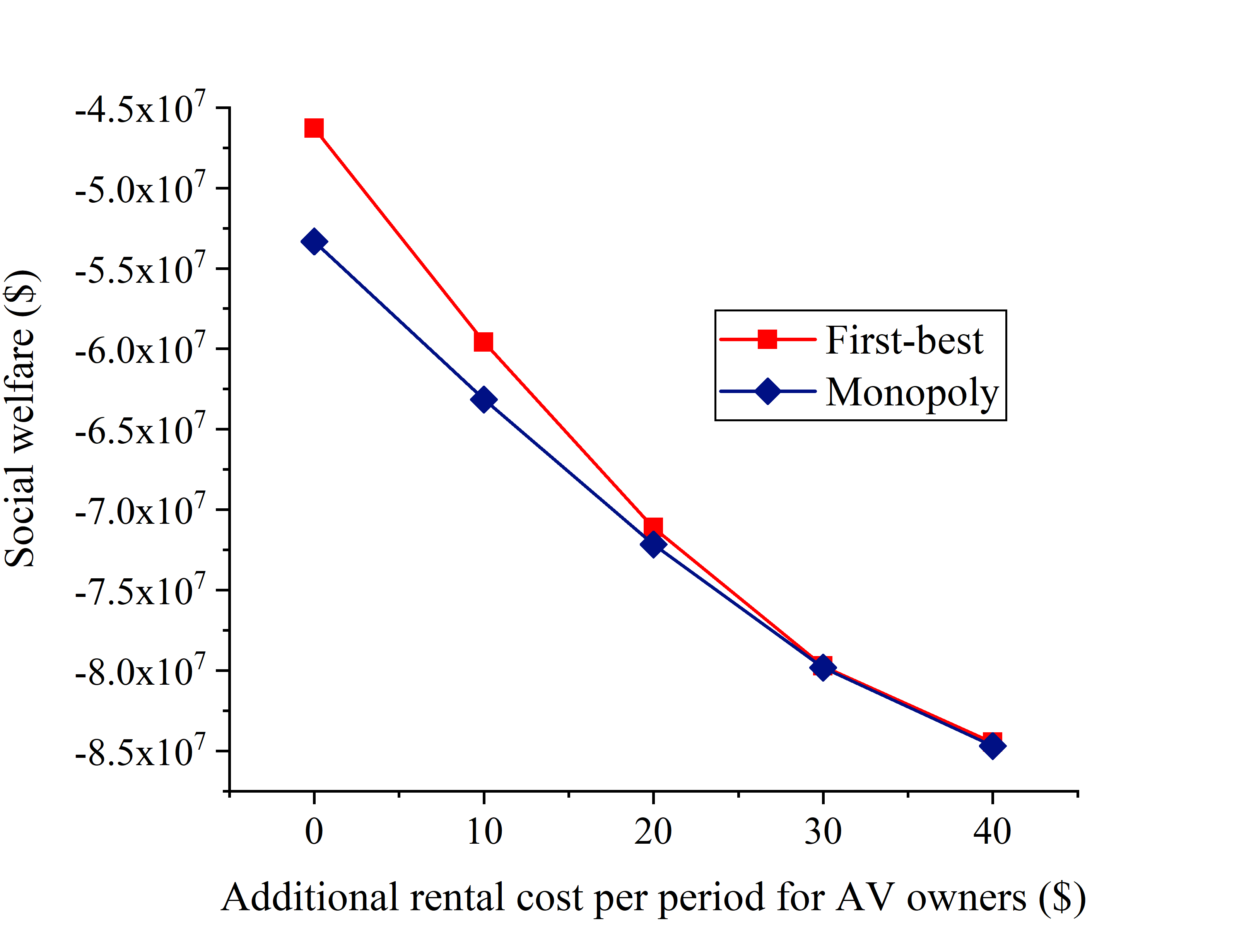}}\\
	\subfloat[][]{\includegraphics[width=0.5\textwidth]{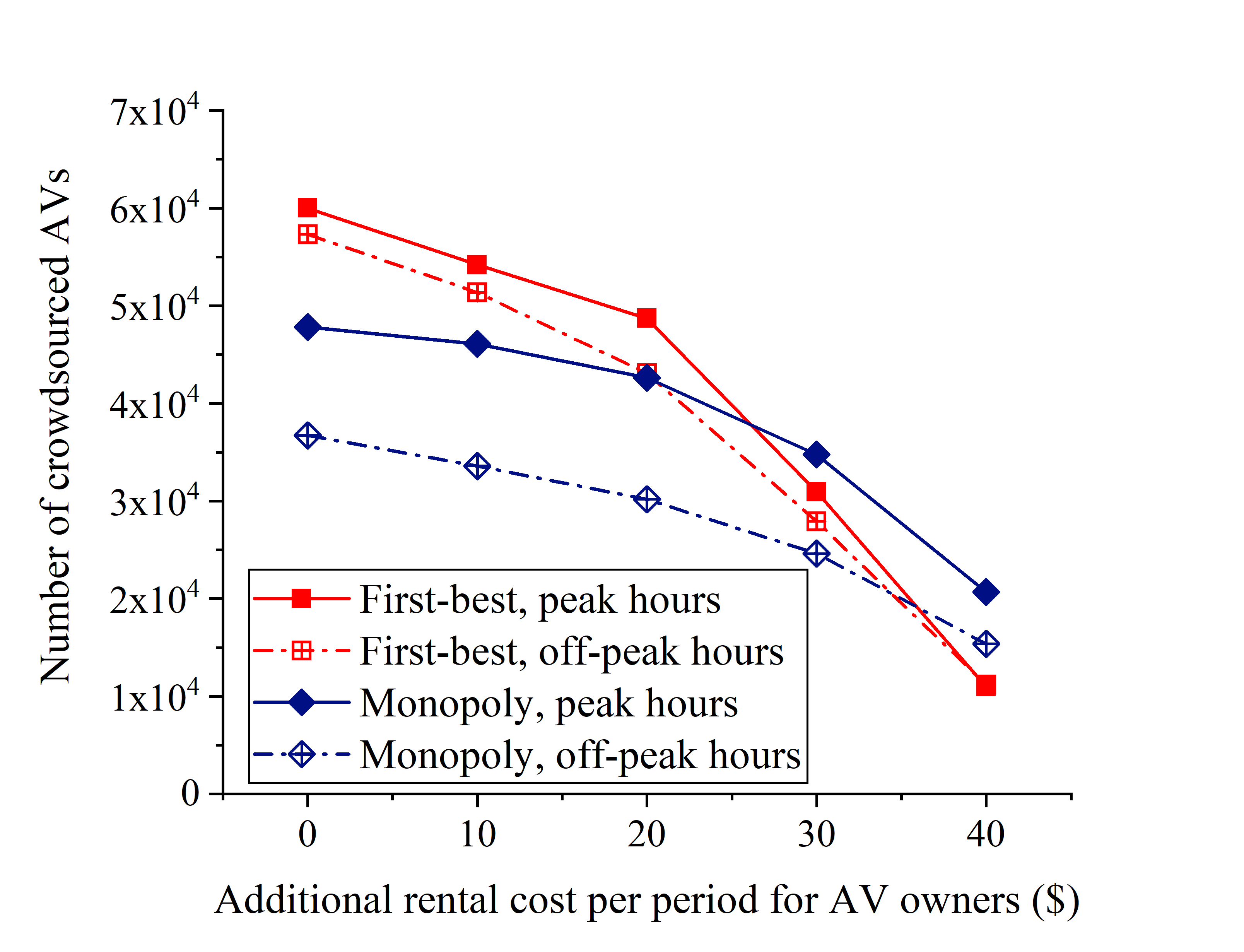}}\\
	\caption[]{Influence of the additional cost to rent out private autonomous vehicles.} 
	\label{fig_psy}
\end{figure}
\par

\subsubsection{Owners' additional costs for sharing AVs}

Figures \ref{fig_psy}(a) and \ref{fig_psy}(b) present the rise of fare and payment with increasing additional cost. The crowdsourcing platform is more difficult to collect AVs and thus must provide greater compensation to the AV owners. The decline of supply leads to a decreasing number of crowdsourcing trips. One noticeable difference is that the payment during the peak hours for the first-best scenario first decreases and then increases. In our model, travel time and additional cost are two components of social welfare. When additional cost is low, the payment in peak hours decreases to discourage crowdsourcing trips for less congestion; however, when additional cost continues to grow, additional cost becomes more dominant for social welfare and thus, the payment increases to compensate the AV owners. An examination of the results presented in Figure \ref{fig_psy}(c) shows that although the monopoly profits decrease monotonically with additional cost $m$, profits in the first-best scenario first increase and then decrease after $m$ exceeds $\$30$. The first half of the curve suggests that under the first-best case, a higher $m$ means a higher profit for the crowdsourcing platform. However, when $m$ is quite high, the first-best scenario and the monopoly show little difference in social welfare and platform's profits. These tests suggest that it is important for the crowdsourcing platform to investigate the AV owners' additional cost prior to setting prices; it may also be important for the urban planner to determine $m$ when deciding whether to impose regulations because a small $m$ means that regulation has a more beneficial effect on the society, as shown in Figure \ref{fig_psy}(d).

\section{Conclusions} \label{sec6}

This paper investigates the AV crowdsourcing market that may emerge in the future. We first proposed an equilibrium model with multiple transport modes: AV on-demand ride service, private AV, private MV and public transit, and proved the existence of the equilibrium solution. Then, several scenarios, namely, the monopoly scenario, the first-best case and the trade-off second-best case were explored.

We use numerical examples to illustrate some insightful findings concerning the AV crowdsourcing market, including the comparison of optimal pricing for trip fare and payment during the peak and off-peak hours, the crowdsourcing platform's profits, the social welfare and transport modes in each scenario. Sensitivity analysis of the optimal pricing and the  equilibrium state with regard to different parameters in the model are presented in a thorough manner. The main findings from the numerical examples are summarized below.\par
\noindent \ 1. AV crowdsourcing service is beneficial for the society overall. This provides guidance for urban planners to have a supportive attitude toward AV development and popularization.\par
\noindent \ 2. Renting out private AVs to the crowdsourcing platform provides high return, encouraging people to rent their cars instead of driving them themselves. This verifies a new opportunity for adventurous investors such as the wealthy to invest into AVs.  \par
\noindent \ 3. Regulation may be easier for the crowdsourcing platform to accept in some second-best scenario where the social welfare is close to that of the first-best scenario. Regulation is particularly necessary when people are not sensitive to utilities, or the AV owners are more inclined to rent out their AVs to gain profits.  \par
\noindent \ 4. It is important for the crowdsourcing platform to know about the people's sensitivity to utilities and AV owners' additional cost prior to setting prices, because these two factors significantly affect the optimal pricing. Additionally, the crowdsourcing platform should expect to be launched in areas with high population density and should purchase as many AVs in advance as possible to obtain higher profits. Nevertheless, even without pre-purchased vehicles, the platform is highly likely to gain profits under the premise that the market penetration rate of AVs is sufficiently high (above 10\%). \par

This paper could be further advanced by considering competition between MV and AV on-demand ride service in the mobility market, as is closer to the reality before the society is fully accustomed to AV technologies. Additionally, as shared autonomous vehicles will probably be operated together with public transit, the joint planning and operation of two modes to form an integrated mobility-as-a-service (MaaS) system should be explored in the future. In terms of wider implications, our results suggest that with the presence of AV crowdsourcing business, even when most AVs are owned privately, these AVs will be mostly used for serving other customers in the society, which motivates the discussion of AV ownership in the future, including where it will converge and how it will evolve. On the other hand, for these AV owners, AVs will not be purchased merely for personal usage, but also for investment (buying vehicles to earn money from day to day), and the capital value of AVs merits careful investigation.

\section*{Acknowledgements}
The research is supported in part by Tsinghua-Daimler Joint Research Center for Sustainable Transportation.

\newpage
\section*{References}
\noindent

Arnott, R. (1996). Taxi travel should be subsidized. Journal of Urban Economics, 40(3):316–333. \par
Bang, S. and Ahn, S. (2018). Control of connected and autonomous vehicles with cut-in movement using spring mass damper system. Transportation Research Record: Journal of the Transportation Research Board, 2672:036119811879692.\par
Bayus, B. (2010). Crowdsourcing and individual creativity over time: The detrimental effects of past success. SSRN Electronic Journal.\par
Bayus, B. L. (2013). Crowdsourcing new product ideas over time: An analysis of the dell ideastorm community. Manage. Sci., 59(1):226–244.\par
Ben-Akiva, M. E. and Lerman, S. R. (1985). Discrete Choice Analysis: Theory and Application to Travel Demand.\par
Bian, Y., Zheng, Y., Ren, W., Li, S. E., Wang, J., and Li, K. (2019). Reducing time headway for platooning of connected vehicles via v2v communication. Transportation Research Part C: Emerging Technologies, 102:87 – 105.\par
Brabham, D. C. (2008). Crowdsourcing as a model for problem solving: An introduction and cases. Convergence, 14(1):75–90.\par
Chen, X.-D., Li, M., Lin, X., Yin, Y., and He, F. (2020). Rhythmic control of automated trafﬁc – part i: Concept and properties at isolated intersections. arXiv: Optimization and Control.\par
Daganzo, C. F. and Ouyang, Y. (2019). A general model of demand-responsive transportation services: From taxi to ridesharing to dial-a-ride. Transportation Research Part B: Methodological, 126:213–224.\par
Estelle´s-Arolas, E. and Gonza´lez-Ladro´n-De-Guevara, F. (2012). Towards an integrated crowdsourcing deﬁnition. J. Inf. Sci., 38(2):189–200.\par
Fagnant, D. J. and Kockelman, K. (2015). Preparing a nation for autonomous vehicles: opportunities, barriers and policy recommendations. Transportation Research Part A: Policy and Practice, 77:167 – 181.\par
Glaeser, E. L., Hillis, A., Kominers, S. D., and Luca, M. (2016). Crowdsourcing city government: Using tournaments to improve inspection accuracy. American Economic Review, 106(5):114–18.\par
Gong, S. and Du, L. (2018). Cooperative platoon control for a mixed trafﬁc ﬂow including human drive vehicles and connected and autonomous vehicles. Transportation Research Part B: Methodological, 116:25 – 61.\par
Gray, J. (2015). Implicit Functions, pages 259–269.\par
He, F., Wang, X., and Tang, X. (2018). Pricing and penalty/compensation strategies of a taxi-hailing platform. Transportation Research Part C Emerging Technologies, 86:263–279.\par
Kang, J., Kuznetsova, P., Luca, M., and Yejin, C. (2013). Where not to eat? improving public policy by predicting hygiene inspections using online reviews.\par
Lee, J. and Park, B. (2012). Development and evaluation of a cooperative vehicle  intersection control algorithm under the connected vehicles environment. IEEE Transactions on Intelligent Transportation Systems, 13(1):81–90.\par
Li, Z., Elefteriadou, L., and Ranka, S. (2014).  Signal control optimization for automated vehicles at isolated signalized intersections. Transportation Research Part C: Emerging Technologies, 49:1 – 18.\par
Lin, X., Li, M., Shen, Z., Yin,  Y.,  and He, F.  (2021). Rhythmic control of automated trafﬁc – part  ii: Grid network rhythm and online routing. Transportation Science, forthcoming.\par
Litman, T. (2015). Autonomous vehicle implementation predictions: Implications for transport planning.\par
Luo, Y., Xiang, Y., Cao, K., and Li, K. (2016). A dynamic automated lane change maneuver based on vehicle-to-vehicle communication. Transportation Research Part C: Emerging Technologies, 62:87– 102.\par
Michael, J. B., Godbole, D. N., Lygeros, J., and Sengupta, R. (1998). Capacity analysis of trafﬁc ﬂow over a single-lane automated highway system. Journal of Intelligent Transportation System, 4(1-2):49–80.\par
Narayanan, S., Chaniotakis, E., and Antoniou, C. (2020). Shared autonomous vehicle services: A comprehensive review. Transportation Research Part C: Emerging Technologies, 111:255 – 293.\par
Poetz, M. K. and Schreier, M. (2012). The value of crowdsourcing: Can users really compete with professionals in generating new product ideas? Journal of Product Innovation Management, 29(2):245–256.\par
Schroeter, J. R. (1983). A model of taxi service under fare structure and ﬂeet size regulation. The Bell Journal of Economics, pages 81–96.\par
Stocker, A. and Shaheen, S. (2018). Shared automated mobility: Early exploration and potential impacts. In Meyer, G. and Beiker, S., editors, Road Vehicle Automation 4, pages 125–139, Cham. Springer International Publishing.\par
Sun, J., Zheng, Z., and Sun, J. (2020). The relationship between car following string instability and trafﬁc oscillations in ﬁnite-sized platoons and its use in easing congestion via connected and automated vehicles with idm based controller. Transportation Research Part B: Methodological, 142:58 – 83.\par
Tang, X., Li, M., Lin, X., and He, F. (2020). Online operations of automated electric taxi ﬂeets: An advisor-student reinforcement learning framework. Transportation Research Part C: Emerging Technologies, 121:102844.\par
Tientrakool, P., Ho, Y.-C., and Maxemchuk, N. F. (2011). Highway capacity beneﬁts from using vehicle-to-vehicle communication and sensors for collision avoidance. In 2011 IEEE Vehicular Technology Conference (VTC Fall), pages 1–5. IEEE.\par
Vander Schee, B. (2009). Crowdsourcing: Why the power of the crowd is driving the future of business [book review] 2009 jeff howe. new york, ny: Crown business. Journal of Consumer Marketing, 26:305–306.\par
Vignon, D. and Yin, Y. (2020). Regulating ride-sourcing services with product differentiation and congestion externality. Available at SSRN 3531372.\par
Vosooghi, R., Puchinger, J., Jankovic, M., and Vouillon, A. (2019). Shared autonomous vehicle simulation and service design. Transportation Research Part C: Emerging Technologies, 107:15 – 33.\par
Wang, C., Gong, S., Zhou, A., Li, T., and Peeta, S. (2019). Cooperative adaptive cruise control for connected autonomous vehicles by factoring communication-related constraints. Transportation Research Procedia, 38:242 – 262. Journal of Transportation and Trafﬁc Theory.\par
Wang, H. and Yang, H. (2019). Ridesourcing systems: A framework and review. Transportation Research Part B: Methodological, 129:122 – 155.\par
Wang, M., Daamen, W., Hoogendoorn, S. P., and van Arem, B. (2014). Rolling horizon control framework for driver assistance systems. part ii: Cooperative sensing and cooperative control. Transportation Research Part C: Emerging Technologies, 40:290 – 311.\par
Wu, J., Ah, S., Zhou, Y., Liu, P., and Qu, X. (2020). The cooperative sorting strategy for connected and automated vehicle platoons. arXiv preprint arXiv:2003.06481.\par
Xu, B., Li, S. E., Bian, Y., Li, S., Ban, X. J., Wang, J., and Li, K. (2018). Distributed conﬂict-free cooperation for multiple connected vehicles at unsignalized intersections. Transportation Research Part C: Emerging Technologies, 93:322 – 334.\par
Yang, D., Zheng, S., Wen, C., Jin, P. J., and Ran, B. (2018). A dynamic lane-changing trajectory planning model for automated vehicles. Transportation Research Part C: Emerging Technologies, 95:228 – 247.\par
Yang, H. and Wong, S. (1998). A network model of urban taxi services. Transportation Research  Part B: Methodological, 32(4):235 – 246.\par
Yang, T., Yang, H., Wong, S. C., and Sze, N. N. (2014). Returns to scale in the production of taxi services: an empirical analysis. Transportmetrica A: Transport Science, 10(9):775–790.\par
Yu, C., Feng, Y., Liu, H. X., Ma, W., and Yang, X. (2018). Integrated optimization of trafﬁc signals and vehicle trajectories at isolated urban intersections. Transportation Research Part B: Methodological, 112:89 – 112.\par
Yu, J. J., Tang, C. S., Shen, Z.-J. M., and Chen, X. (2017). Should on-demand ride services be regulated? an analysis of chinese government policies. An Analysis of Chinese Government Policies (June 16, 2017).\par
Zervas, G., Proserpio, D., and Byers, J. W. (2017). The rise of the sharing economy: Estimating the impact of airbnb on the hotel industry. Journal of Marketing Research, 54(5):687–705.\par
Zha, L., Yin, Y., and Yang, H. (2016). Economic analysis of ride-sourcing markets. Transportation Research Part C: Emerging Technologies, 71:249–266.\par
Zhao, Y. and Zhu, Q. (2014). Evaluation on crowdsourcing research: Current status and future direction. Information Systems Frontiers, 16(3):417–434.\par
Zmud, J., Williams, T., Outwater, M., Bradley, M., Kalra, N., and Row, S. (2018). Updating regional transportation planning and modeling tools to address impacts of connected and automated vehicles, volume 1: Executive summary. Technical report.\par
Zou, H. (2008). Chapter 7 - ﬁxed point theory and elliptic boundary value problems. In Chipot, M., editor, Handbook of Differential Equations, volume 6 of Handbook of Differential Equations: Stationary Partial Differential Equations, pages 503 – 583. North-Holland.\par

\newpage
\section*{Appendix A. Nomenclature}
\begin{figure}[]
	\centering
	\includegraphics[width=1\textwidth]{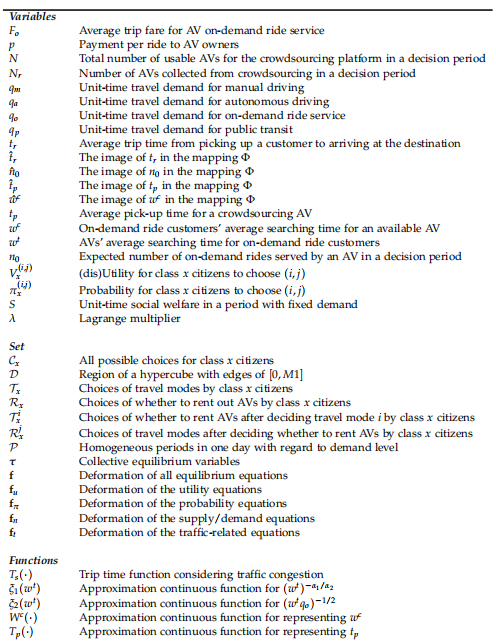}
\end{figure}
\newpage
\begin{figure}[]
	\centering
	\includegraphics[width=1\textwidth]{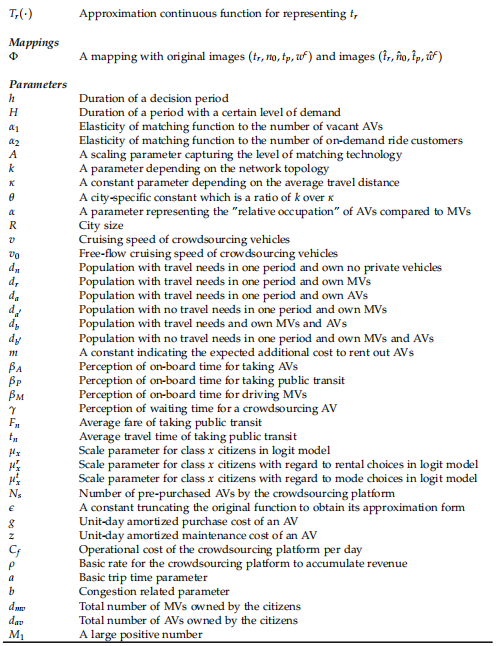}
\end{figure}

\newpage
\section*{Appendix B. Discussion about logit model in Section \ref{sec3_2}}
\renewcommand{\theequation}{B-\arabic{equation}}
\setcounter{equation}{0}

Let $\mu_x^r$ denote the scale parameter of logit model with regard to rental choices; $\mu_x^t$ denote the scale parameter of logit model with regard to transport mode choices.
If $\mu_{x}^{r}>\mu_{x}^{t}$, citizens first choose travel modes and then choose whether to rent their AVs to the crowdsourcing platform. Let 
$\mathcal{T}_{x}$ denote the choice set of travel modes by class $x$ citizens and $\mathcal{T}_{x}^{i}$ denote the rental choice set after they have chosen travel mode $i$. Therefore, the probability of class $x$ citizens choosing $(i,j)\in \mathcal{C}_x$ can be obtained using the nested logit model:

\begin{align}
& \pi_x^{(i,j)}=\frac{e^{\mu_{x}^{r} V_x^{(i,j)}}}{\sum\nolimits_{j \in \mathcal{T}_{x}^{i}} e^{\mu_{x}^{r} V_x^{(i,j)}}}  \frac{{(\sum\nolimits_{j \in \mathcal{T}_{x}^{i}}^{} e^{\mu_{x}^{r} V_x^{(i,j)}})}^{\mu_{x}^{t}/\mu_{x}^{r}}}{\sum\nolimits_{i \in \mathcal{T}_x}^{} {(\sum\nolimits_{j \in \mathcal{T}_{x}^{i}}^{} e^{\mu_{x}^{r} V_x^{(i,j)}})}^{\mu_{x}^{t}/\mu_{x}^{r}}}
& \forall (i,j) \in \mathcal{C}_x, x \in \{1,2,\dots,6\}\label{eqB1}
\end{align}
\noindent

If $\mu_{x}^{r}<\mu_{x}^{t}$, citizens first choose whether to rent their AVs to the crowdsourcing platform and then choose travel modes. Let 
$\mathcal{R}_{x}$ denote the rental choice set by class $x$ citizens and $\mathcal{R}_{x}^{j}$ denote the choice set of travel modes after they have made rental choice $j$. Therefore, the probability of class $x$ citizens choosing $(i,j)\in \mathcal{C}_x$ can be obtained using the nested logit model:

\begin{align}
& \pi_x^{(i,j)}=\frac{e^{\mu_{x}^{t} V_x^{(i,j)}}}{\sum\nolimits_{i \in \mathcal{R}_{x}^{j}} e^{\mu_{x}^{t} V_x^{(i,j)}}}  \frac{{(\sum\nolimits_{i \in \mathcal{R}_{x}^{j}}^{} e^{\mu_{x}^{t} V_x^{(i,j)}})}^{\mu_{x}^{r}/\mu_{x}^{t}}}{\sum\nolimits_{j \in \mathcal{R}_x}^{} {(\sum\nolimits_{i \in \mathcal{R}_{x}^{j}}^{} e^{\mu_{x}^{t} V_x^{(i,j)}})}^{\mu_{x}^{r}/\mu_{x}^{t}}}
& \forall (i,j) \in \mathcal{C}_x, x \in \{1,2,\dots,6\}\label{eqB2}
\end{align}

If $\mu_{x}^{r}=\mu_{x}^{t}=\mu_x$, people decide simultaneously whether to rent their AVs and the  type of travel mode they take. Therefore, the probability of class $x$ citizens choosing $(i,j)\in \mathcal{C}_x$ is reduced to the logit model:

\begin{align}
& \pi_x^{(i,j)}=\frac{e^{\mu_x V_x^{(i,j)}}}{\sum\nolimits_{(i,j) \in \mathcal{C}_x}^{} e^{\mu_x V_x^{(i,j)}}} & \forall (i,j) \in \mathcal{C}_x, x \in \{1,2,\dots,6\}\label{eqB3}
\end{align}

\end{document}